\documentclass[11pt]{article}
\usepackage{fullpage}

\usepackage[numbers,sort&compress]{natbib}

\usepackage{dirtytalk}
\usepackage{amsthm,amsmath,amssymb}

\usepackage{mathtools}

\usepackage{etoolbox}

\usepackage{hyperref}
\hypersetup{
  hidelinks,
  colorlinks=true,       %
  raiselinks=true,       %
  allcolors=black,
  pdfstartview=Fit,
  breaklinks=true,       %
  hypertexnames=false,   %
}

\usepackage[all]{hypcap}

\PassOptionsToPackage{table,xcdraw,svgnames,dvipsnames}{xcolor}

\usepackage{xltabular}

\usepackage{url}
\usepackage{microtype}
\usepackage{relsize}
\usepackage{listings}
\usepackage{xspace}
\usepackage{multirow}

\usepackage{thmtools}
\usepackage{thm-restate}
\usepackage{booktabs}       %
\usepackage{microtype}      %
\usepackage{colortbl}       %
\usepackage{stackengine}
\usepackage{dsfont}
\usepackage{tabularx}
\usepackage{float}
\usepackage{mathtools}
\usepackage{mathrsfs}
\usepackage{xspace}
\usepackage{xparse}
\usepackage[export]{adjustbox}
\usepackage[all]{nowidow}

\usepackage{conf/tikzit}
\usepackage{tikz}

\usepackage{graphics}
     \graphicspath{{fig/}}
\usepackage{graphicx}
\usepackage[style=altlist,acronym]{glossaries}
\usepackage{pgf}

\tikzstyle{BLACK}=[draw=black, shape=circle, fill=black, inner sep=2.6pt]
\tikzstyle{DOM}=[fill={rgb,255: red,122; green,0; blue,42}, draw=black, shape=circle, inner sep=3pt, line width=0.5px]
\tikzstyle{NONE}=[fill={rgb,255: red,239; green,66; blue,109}, draw=black, shape=circle, inner sep=3pt, line width=0.5px]
\tikzstyle{NTWO}=[fill={rgb,255: red,111; green,176; blue,118}, draw=black, shape=circle, inner sep=3pt, line width=0.5px]
\tikzstyle{NTHR}=[fill={rgb,255: red,252; green,199; blue,18}, draw=black, shape=circle, inner sep=3pt, line width=0.5px]
\tikzstyle{GRAYN}=[fill={rgb,255: red,213; green,213; blue,213}, draw=black, shape=circle, inner sep=3pt]
\tikzstyle{BRGRAY}=[fill={rgb,255: red,236; green,236; blue,236}, draw=none, shape=circle, inner sep=3pt]
\tikzstyle{TEXTSTD}=[fill=white, draw=black, shape=rectangle, tikzit shape=rectangle, text width=3cm, rounded corners]
\tikzstyle{TEXTHIGH}=[fill={rgb,255: red,231; green,245; blue,255}, draw=black, shape=rectangle, tikzit shape=rectangle, text width=3cm, rounded corners]
\tikzstyle{TEXTHIGHGRAY}=[fill={rgb,255: red,220; green,220; blue,220}, draw=black, shape=rectangle, tikzit shape=rectangle, text width=3cm, rounded corners]
\tikzstyle{TEXTHIGHGREEN}=[fill={rgb,255: red,149; green,242; blue,145}, draw=black, shape=rectangle, tikzit shape=rectangle, text width=3.4cm, rounded corners]
\tikzstyle{TEXTHIGHYELLOW}=[fill={rgb,255: red,250; green,250; blue,160}, draw=black, shape=rectangle, tikzit shape=rectangle, text width=3cm, rounded corners]

\tikzstyle{EDGE}=[-, fill=none, line width=0.90px]
\tikzstyle{BLUE}=[-, draw={rgb,255: red,0; green,101; blue,189}, line width=1.5px]
\tikzstyle{GRAY}=[-, fill={rgb,255: red,128; green,128; blue,128}]
\tikzstyle{DARKGREEN}=[-, fill=none, draw={rgb,255: red,19; green,100; blue,13}, line width=1.25pt]
\tikzstyle{big dash}=[-, dashed, dash pattern=on 4mm off 2mm, fill={rgb,255: red,178; green,255; blue,253}]
\tikzstyle{big dash thick}=[-, thick, dashed, dash pattern=on 4mm off 2mm, fill={rgb,255: red,178; green,255; blue,253}]
\tikzstyle{new edge style 0}=[-, fill={rgb,255: red,230; green,255; blue,252}]
\tikzstyle{FILLGREEN}=[-, fill={rgb,255: red,232; green,255; blue,207}]
\tikzstyle{FILLBLUE}=[-, fill={rgb,255: red,247; green,249; blue,255}, line width=0.9px]
\tikzstyle{FILLPURPLE}=[-, fill={rgb,255: red,255; green,233; blue,239}]
\tikzstyle{FILLDARKBLUE}=[-, fill={rgb,255: red,238; green,230; blue,255}]
\tikzstyle{FILLDARKBLUEINVISIBLE}=[-, draw=none, fill={rgb,255: red,238; green,230; blue,255}]
\tikzstyle{FILLGREENINVISBLE}=[-, draw=none, fill={rgb,255: red,232; green,255; blue,207}]
\tikzstyle{ULTRAGRAY}=[-, draw={rgb,255: red,213; green,213; blue,213}, line width=0.9px]
\tikzstyle{EDGEDASHED}=[-, dashed, fill=none, line width=0.9px]
\tikzstyle{simple directed edge}=[->, draw=black, thick, line width=0.45mm]
\tikzstyle{simple dashed directed edge}=[->, draw=black, dashed]
\tikzstyle{SLIGHTLYDASHED}=[-, dotted, fill=none, draw={rgb,255: red,128; green,128; blue,128}]
\tikzstyle{simple dashed directed edge}=[->, draw=black, line width=0.75mm, dashed]
\tikzstyle{simple dashed directed edge not thick}=[->, draw=black, line width=0.45mm, dashed]
\tikzstyle{simple}=[->, draw=black, thick]
\tikzstyle{simple}=[draw=black, fill=none, tikzit draw=black, ->]
\tikzstyle{invis}=[-, draw=none]
\tikzstyle{g5edge}=[-, line width=1pt]
\tikzstyle{DashedBlue}=[-]
\tikzstyle{new edge style 1}=[-, fill=none, draw={rgb,255: red,0; green,99; blue,185}, line width=1pt, dashed]

 \tikzset{
  every node/.append style={text=black},
  every label/.append style={text=black}
}
\usetikzlibrary{positioning,chains,fit,shapes,calc,er,automata}

\usepackage[framemethod=TikZ]{mdframed}
\usepackage{footnote}
\usepackage{booktabs}
\usepackage{soul}
\usepackage{arydshln}

\usepackage[inline]{enumitem}
\usepackage{cleveref}

\usepackage{dsfont}

\newcommand*\ie{i.\kern.1em e., }
\newcommand*\eg{e.\kern.1em g., }
\newcommand*\cf{c.\kern.1em f.\ }
\newcommand*\almev{a.\kern.1em e.\ }

\theoremstyle{plain}
\newtheorem{theorem}{Theorem}[section]
\newtheorem{definition}{Definition}[section]
\newtheorem{lemma}{Lemma}[section]
\newtheorem{corollary}{Corollary}[section]
\newtheorem{proposition}{Proposition}[section]
\newtheorem{observation}{Observation}

\crefname{definition}{Definition}{definitions}
\Crefname{definition}{Definition}{definitions}

\crefname{observation}{Observation}{Observations}
\Crefname{observation}{Observation}{Observations}

\crefname{claim}{Claim}{Claims}
\Crefname{claim}{Claim}{Claims}

\crefname{fact}{Fact}{Facts}
\Crefname{fact}{Fact}{Facts}

\crefname{theorem}{Theorem}{Theorems}
\Crefname{theorem}{Theorem}{Theorems}

\crefname{lemma}{Lemma}{lemmas}
\Crefname{lemma}{Lemma}{lemmas}

\crefname{corollary}{Corollary}{Corollaries}
\Crefname{corollary}{Corollary}{Corollaries}

\crefname{figure}{Figure}{Figures}
\Crefname{figure}{Figure}{Figures}

\newcommand{\ignore}[1]{}

\DeclarePairedDelimiter{\abs}{\lvert}{\rvert}

\newcommand{\cF}{\ensuremath{\mathcal{F}}}

\newcommand{\cI}{\ensuremath{\mathcal{I}}}

\newcommand{\cO}{\ensuremath{\mathcal{O}}}

\newcommand{\fR}{\ensuremath{\mathfrak{R}}}

\numberwithin{equation}{section}

\colorlet{linkcolour}{DarkBlue}

\hypersetup{colorlinks=true, linkcolor=linkcolour, citecolor=linkcolour, urlcolor=linkcolour,}

\renewcommand{\epsilon}{\varepsilon}

\setul{}{1.5pt} %

\newcommand{\allbold}[1]{%
  {\bfseries\boldmath #1}%
}

\usepackage{bbold}

\definecolor{TUMBlue}{HTML}{0065BD}
\definecolor{TUMTertiaryBlue}{HTML}{0063B9}
\definecolor{TUMSecondaryBlue}{HTML}{005293}
\definecolor{TUMSecondaryBlue2}{HTML}{003359}
\definecolor{TUMBlack}{HTML}{000000}
\definecolor{TUMWhite}{HTML}{FFFFFF}
\definecolor{TUMDarkGray}{HTML}{333333}
\definecolor{TUMGray}{HTML}{808080}
\definecolor{TUMLightGray}{HTML}{CCCCC6}
\definecolor{TUMAccentGray}{HTML}{DAD7CB}
\definecolor{TUMAccentOrange}{HTML}{E37222}
\definecolor{TUMAccentGreen}{HTML}{A2AD00}
\definecolor{TUMAccentLightBlue}{HTML}{98C6EA}
\definecolor{TUMAccentBlue}{HTML}{64A0C8}
\definecolor{VeryLightBlue}{HTML}{E8F8FF}
\definecolor{NONE}{HTML}{F721D4}
\definecolor{NTWO}{HTML}{00B1DB}
\definecolor{NTHREE}{HTML}{8C8F0E}
\definecolor{MATHAGREEN}{HTML}{B8E986}
\definecolor{MATHARED}{HTML}{F3B1B1}
\definecolor{MATHABLUE}{HTML}{CBB3FF}
\definecolor{MATHAYELLOW}{HTML}{EEE46F}
\definecolor{FILLGREEN}{HTML}{E8FFCF}
\definecolor{FILLPURPLE}{HTML}{E8FFCF}
\definecolor{FILLDARKBLUE}{HTML}{EEE6FF}

\definecolor{darkred}{HTML}{E80000}
\definecolor{darkgreen}{HTML}{006400} 
\definecolor{yellowgreen}{HTML}{9ACD32}

\definecolor{sbOrange}{HTML}{D78355}
\definecolor{sbGreen}{HTML}{60AD72}
\definecolor{sbRed}{HTML}{C7595C}
\definecolor{sbBlue}{HTML}{4C72B0}

    \definecolor{nicepurple}{HTML}{802c7f}
    \definecolor{orangeish}{HTML}{f1a340} %
    \definecolor{grayish}{HTML}{f7f7f7} %
    \definecolor{purpleish}{HTML}{998ec3} %
    \definecolor{blueish}{HTML}{004488}
    \definecolor{darkgreen}{RGB}{2,100,64} %
    \definecolor{lightgreen}{HTML}{b2f2bb}
    \definecolor{lightblueish}{RGB}{230,244,255}
\PassOptionsToPackage{table,xcdraw,svgnames,dvipsnames}{xcolor}

\makeatletter
\g@addto@macro\thmt@innercounters{,casenum}
\makeatother

\newcounter{casenum}

\AtBeginDocument{%
  }

\newcommand\kernelsize{358}

\usepackage{cleveref}
\newcommand{\name}[1]{\textsc{#1}}
\newcommand{\dom}	{\name{Domination}\xspace}

\newcommand{\tdom}{\name{Total Domination}\xspace}

\newcommand{\sdom}{\name{Semitotal Domination}\xspace}

\newcommand{\dreg}{\textit{$D$-region decomposition}\xspace}

\newcommand{\NPcs}{\ensuremath{\name{NP}}-c\xspace}

\newcommand{\WitRad}{\ensuremath{W_G}\xspace}

\newcommand{\Ptt}{\name{P}\xspace}
\newcommand{\FPT}{\name{FPT}\xspace}
\newcommand{\FPTt}{\name{FPT}\xspace}

\newcommand{\WTWOhs}{\ensuremath{W[2]}-h\xspace}
\newcommand{\WONEhs}{\ensuremath{W[1]}-h\xspace}

\newcommand{\G}		{\ensuremath{G=(V,E)}\xspace}
\newcommand{\GB}	{\ensuremath{G'=(V',E')}\xspace}

\newcommand{\sr}	{simple region}

\newcommand{\dGR} {\ensuremath{d_{G_{\mathfrak{R}}}}}
\newcommand{\GR} {\ensuremath{{G_{\mathfrak{R}}}}}

\newcommand{\nthi}{\ensuremath{\abs{N_3(v,w) \cap V(R)}}}

\newcommand{\mc}{\mathcal}

\newcommand{\Dvw}   {\ensuremath{\mc{D}}\xspace}
\newcommand{\Dv}    {\ensuremath{\mc{D}_{v}}\xspace}
\newcommand{\Dw}    {\ensuremath{\mc{D}_{w}}\xspace}

\newcommand{\Nthreev}    {\ensuremath{N_3(v)}}

\newcount\notes%
    
\newcommand{\revise}[1]{{#1}}
    \newcommand{\revisenote}[1]{}

\ifnum\notes=1
    \newcommand{\lr}[1]{{\color{Navy} #1}}
    \newrobustcmd{\lrnote}[1]{%
          {\color{Navy}\footnote{{\color{DarkGreen}\textbf{LR:} #1}}}}%

\else 
    \newcommand{\lr}[1]{{#1}}
    \newcommand{\lrnote}[1]{}

\fi

{\bfseries}{\itshape}
\newtheorem{rgl}{Rule}{\bfseries}{\itshape}

\usepackage[T1]{fontenc}

\newif\ifanon
\anonfalse

\usepackage[labelfont=bf,font=small,skip=4pt]{caption}
\usepackage{graphicx}

\begin{document}
\let\WriteBookmarks\relax
\def\floatpagepagefraction{1}
\def\textpagefraction{.001}

\title{A Linear Kernel for \revise{Planar} Semitotal Domination}

\author {
  Paloma T. de Lima\\
  Norwegian School of Economics\\
  \texttt{Paloma.Lima@nhh.no}
  \and
  Lukas Retschmeier\\
  BARC, University of Copenhagen\\
  \texttt{retschmeier.lukas@gmail.com}
}

\maketitle
\begin{abstract}
    A set of vertices $S$ in a graph $G$ is a  
    \emph{semitotal dominating set} if $S$ is a dominating set and every vertex in $S$ has a witness, 
    that is, another vertex in $S$ within distance at most two.
    The \name{Semitotal Domination} problem asks whether $G$ admits such a set of size at most~$k$.
    We study this problem parameterized by the solution size.
    We provide a linear kernel of size $358k$ on planar graphs, a result that adds to several existing linear kernels for other domination variants on planar graphs.
    Additionally, we show \sdom is $W[2]$-hard on split graphs, on bipartite graphs and on $t$-claw-free graphs, it is $W[1]$-hard on circle graphs.%

\end{abstract}

\section{Introduction}
\noindent Domination problems play a central role in graph theory and parameterized algorithms.
A \emph{dominating set} in a graph is a subset $S$ of its vertices such that every vertex not in $S$ has at least one neighbor in $S$. 
In the classical \dom problem, the input is a graph $G$ and an integer $k$, and the question is whether $G$ has a dominating set of size at most $k$. 
Due to its wide applicability, this problem has been extensively studied in the literature (see surveys \cite{Haynes1998,Haynes2020,Haynes2021,Haynes2022,Haynes1998b}) and it has been pivotal in the development of Parameterized Complexity~\cite{Downey1999a,Alber2002}.
As domination showed itself as a concept of importance, several variations of the classical notion were introduced to cater to the specifics of each application (see examples in \cite{Haynes1998}).
Such variations are usually defined by imposing some additional constraint on the dominating set. Examples of these variants include independent domination, connected domination, perfect domination, paired domination and total domination. In this work, we consider the variant of domination known as \emph{semitotal domination}, defined by Goddard, Henning and McPillan~\cite{Goddard2014} in 2014.

A \emph{semitotal dominating set} of a graph $\G$ is a subset $S\subseteq V$ that is a dominating set and, additionally, for each $d \in S$ there exists a vertex $d'\in S \setminus \{d\}$ at distance at most $2$ from $d$. 
We call the vertex $d'$ a \emph{witness} for $d$. 

This concept generalizes that of connected domination (in which we require the dominating set to induce a connected subgraph) and that of total domination (in which every vertex of the dominating set must also have a neighbor in it). 
In the \sdom problem the question is whether an input graph $G$ has a semitotal dominating set of size at most $k$ and we call the minimum $k$ such that a solution exists $\gamma_{t2}(G)$.
\revise{
Throughout this work, we assume that $G$ is connected and has at least two vertices, thereby excluding the trivial case of isolated vertices.
}

Combinatorial properties of $\gamma_{t2}(G)$ have been investigated in \cite{chen2024semitotal,Chen2020,Galby2023,Henning2014matching,henning2017edge,henning2016semitotal,Zhu2019,Aniversario2019} and the problem has also been studied when restricted to \lr{special} graph classes \cite{Tripathi2021,Pradhan2021,Henning2022,Kloks2021,Galby2020}. 
\lr{To the best of our knowledge}, semitotal domination has not yet been studied systematically from the perspective of parameterized complexity. 
This is the focus of this work.

We start by providing a linear kernel for \sdom parameterized by solution size on planar graphs. 
Our kernel is inspired by the foundational linear kernel for \dom on planar graphs due to Alber, Fellows and Niedermeier~\cite{Alber2004}, as well as a more recent analogous result for total domination, obtained by Garnero and Sau~\cite{Garnero2018}. 
This result adds to the substantial body of work providing linear kernels for variations of the domination problem on planar graphs~\cite{Guo2007, Garnero2017, Luo2013, Alber2006,Diekert2005,sahili2025linearkernelplanarvector}. 
More precisely, we prove the following.

\begin{theorem}\label{thm:central}
    The \sdom problem parameterized by solution size admits a linear kernel on planar graphs. More precisely, 
    there exists a polynomial-time algorithm that, given an instance $(G, k)$ of \sdom, either correctly reports that $(G,k)$ is a no-instance, or returns an equivalent instance $(G', k)$ such that $\abs{V(G')} \leq \kernelsize k$.
\end{theorem}

\begin{figure}[b]
    \centering
    \caption{(left) Example of a dominating set with $\gamma(G) = 2$ and (right) a semitotal dominating set with $\gamma_{t2} = 3$ on the same topology.
    Note that vertex $d_2$ is \revise{necessarily} a witness for both $d_1$ and \revise{$d_3$}.}\label{fig:dsexamples}
     \resizebox{0.50\columnwidth}{!}{
          \begin{tikzpicture}
	\begin{pgfonlayer}{nodelayer}
		\node [style=BLACK] (6) at (3.5, 2.75) {};
		\node [style=DOM] (7) at (5, 2.75) {};
		\node [style=BLACK] (8) at (6.5, 2.75) {};
		\node [style=DOM] (9) at (8, 2.75) {};
		\node [style=DOM] (10) at (9.5, 2.75) {};
		\node [style=BLACK] (11) at (11, 2.75) {};
		\node [style=none] (18) at (1.75, 4) {};
		\node [style=none] (19) at (1.75, 1.75) {};
		\node [style=none] (22) at (-3, 4.25) {$\dom$};
		\node [style=none] (23) at (7.25, 4.25) {$\sdom$};
		\node [style=none] (27) at (5, 1.75) {$d_1$};
		\node [style=none] (30) at (8, 1.75) {$d_2$};
		\node [style=none] (33) at (9.5, 1.75) {$d_3$};
		\node [style=BLACK] (35) at (-7, 2.75) {};
		\node [style=DOM] (36) at (-5.5, 2.75) {};
		\node [style=BLACK] (37) at (-4, 2.75) {};
		\node [style=BLACK] (38) at (-2.5, 2.75) {};
		\node [style=DOM] (39) at (-1, 2.75) {};
		\node [style=BLACK] (40) at (0.5, 2.75) {};
		\node [style=none] (41) at (-5.5, 1.75) {$d_1$};
		\node [style=none] (42) at (-1, 1.75) {$d_2$};
	\end{pgfonlayer}
	\begin{pgfonlayer}{edgelayer}
		\draw [style=EDGE] (7) to (6);
		\draw [style=EDGE] (7) to (8);
		\draw [style=EDGE] (9) to (8);
		\draw [style=EDGE] (9) to (10);
		\draw [style=EDGE] (10) to (11);
		\draw [style=EDGEDASHED] (18.center) to (19.center);
		\draw [style=EDGE] (36) to (35);
		\draw [style=EDGE] (36) to (37);
		\draw [style=EDGE] (38) to (37);
		\draw [style=EDGE] (38) to (39);
		\draw [style=EDGE] (39) to (40);
	\end{pgfonlayer}
\end{tikzpicture}
        }
\end{figure}
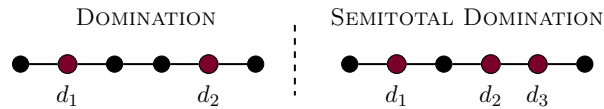

\noindent Finally, we conclude with the following \lr{intractability results.}

\begin{theorem}\label{lemma:hardness}
\sdom parameterized by solution size is $W[2]$-hard even when restricted to bipartite graphs, split graphs, and $t$-claw-free graphs, and it is $W[1]$-hard for circle graphs.
\end{theorem}

\newcommand{\app}{$\clubsuit$}
A summary of the (parameterized) complexity of \dom, \sdom and \tdom~\lr{when restricted to special} graph classes can be found in Table~\ref{tab:complexities}.
\lr{We note that general NP-completeness of \tdom was already mentioned in \cite{Garey1979,zhu2009approximation}  but to the best of our knowledge only \revise{explicitly} proven in \cite{Garnero2018}.
    }

Our paper is organized as follows. Section~\ref{sec:preliminaries} sets the notation and terminology to be used in the rest of the paper. 
Section~\ref{sec:planar} presents our linear kernel for planar graphs (Theorem~\ref{thm:central}), while the proof of \Cref{lemma:hardness}  can be found in Section~\ref{sec:graphclasses}. 
Section~\ref{sec:conclusion} presents open problems for future research. 

\begin{table}[t]
    \centering
    \resizebox{1.0\textwidth}{!}{
   \begin{tabular}{lllllll}
\textbf{Graph Class}    & \multicolumn{2}{c}{\textbf{\dom}}                         & \multicolumn{2}{c}{\textbf{\sdom}}                           & \multicolumn{2}{c}{\textbf{\tdom}}                                   \\ \cline{2-7} 
                        & \textit{Classical}                   & \textit{Parameterized}               & \textit{Classical}                & \textit{Parameterized}                     & \textit{Classical}               & \textit{Parameterized}                              \\
bipartite               & \NPcs~\textsuperscript{\cite{Bertossi1984}}   & \WTWOhs~\textsuperscript{\cite{Raman2008}}    & \NPcs~\textsuperscript{\cite{Henning2019}} & \allbold{\WTWOhs \textsuperscript{Th. \ref{lemma:hardness}}}  & \NPcs~\textsuperscript{\cite{Pfaff1983}}  & {\WTWOhs \textsuperscript{\cite{Krithika2021}}}                                         \\
circle                  & \NPcs~\textsuperscript{\cite{Keil1993}}       & \WONEhs~\textsuperscript{\cite{Bousquet2012}} & \NPcs~\textsuperscript{\cite{Kloks2021}}   & \allbold{\WONEhs~\textsuperscript{Th. \ref{lemma:hardness}}}                                 & \NPcs~\textsuperscript{\cite{Keil1993}}  & \WONEhs~\textsuperscript{\cite{Bousquet2012}}                \\
split                   & \NPcs~\textsuperscript{\cite{Bertossi1984}}   & \WTWOhs~\textsuperscript{\cite{Raman2008}}    & \NPcs~\textsuperscript{\cite{Henning2019}} & \WTWOhs \textsuperscript{\cite{Henning2019}} & \NPcs~\textsuperscript{\textsuperscript{\cite{Laskar1983}}} & \WTWOhs~\textsuperscript{\cite{Raman2008}}                   \\
3-claw-free             & \NPcs~\textsuperscript{\textsuperscript{\cite{hedetniemi1988recent}}}      & \FPTt~\textsuperscript{\cite{Cygan2011}}      & \NPcs~\textsuperscript{\cite{Galby2020}}                        & ?                                 & \NPcs~\textsuperscript{\cite{McRae1995}}  & ?                                          \\
$t$-claw-free, $t>3$    & \NPcs~\textsuperscript{\cite{hedetniemi1988recent}}      & \WTWOhs~\textsuperscript{\cite{Cygan2011}}    & {\NPcs~\textsuperscript{\cite{Galby2020}}}                        & \allbold{$W[2]$-h~\textsuperscript{Lem. \ref{lem:claw-free}}}                                & \NPcs~\textsuperscript{\cite{McRae1995}}  & {\WTWOhs~\textsuperscript{Lem. \ref{lem:claw-free}}}                                         \\
chordal bipartite       & \NPcs~\textsuperscript{\cite{Mueller1987}}    & ?                           & \NPcs~\textsuperscript{\cite{Henning2019}} & ?                                 & \multicolumn{2}{c}{\Ptt~\textsuperscript{\cite{Damaschke1990}}}                        \\
planar                  & \NPcs~\textsuperscript{\cite{Garey1979}}      & \FPTt~\textsuperscript{\cite{Alber2004}}      & \NPcs ~\textsuperscript{\cite{Henning2019}}                   & \allbold{\FPT \textsuperscript{Th. \ref{thm:central}}}        & \NPcs~\textsuperscript{\cite{Garnero2018}}                   & \FPTt~\textsuperscript{\cite{Garnero2018}}                   \\
girth $\geq 5$                  & \NPcs~\textsuperscript{\cite{Zvervich1995}}      & \FPTt~\textsuperscript{\cite{Raman2008}}      & \NPcs~\textsuperscript{\cite{Galby2020}}                     & {\FPT \textsuperscript{\cite{thara2026influence}}}        & \NPcs~\textsuperscript{\cite{Henning2000}}                   & \FPTt~\textsuperscript{\cite{Krithika2021}}                   \\
\hline
bounded clique-width    & \multicolumn{2}{c}{\FPTt~\textsuperscript{\cite{Courcelle1990}}}             & \multicolumn{2}{c}{\FPTt~\textsuperscript{\cite{Courcelle1990}}}                & \multicolumn{2}{c}{\FPTt~\textsuperscript{\cite{Courcelle1990}}}                        \\
bounded mim-width       & \multicolumn{2}{c}{\textsc{XP}~\textsuperscript{\cite{Belmonte2011,BuiXuan2013}}}  & \multicolumn{2}{c}{\textsc{XP}~\textsuperscript{\cite{Galby2020}}}                    & \multicolumn{2}{c}{\textsc{XP}~\textsuperscript{\cite{Belmonte2011,BuiXuan2013}}}             \\ \hline \end{tabular}

    }\caption{Complexities on selected graph classes.
    The parameterization is by solution size.
    \NPcs indicates classic NP-completeness.
    Furthermore, the mim-width results for all problems follow from the meta-theorem by Bergougnoux, Jaffke, and Dreier \cite{Bergougnoux2023}.
    }\label{tab:complexities}
\end{table}

\section{Terminology and Notation}\label{sec:preliminaries}
\noindent \lr{We use the set $[n] := \{1, \cdots, n\}$.}
Let $G = (V(G),E(G))$ be a graph.
We use $n := |V(G)|$ and $m := |E(G)|$ (or simply $V, E$ if the context is clear).
Unless stated otherwise, we consider simple, finite, and undirected graphs with $n\geq2$ but sometimes consider multigraphs, where multiple edges between two vertices are allowed.
\lr{
For two vertices $v,w\in V$, we write $v \sim w$ if $\{v,w\} \in E$.}
For a graph $H$, we call $G$ a \emph{subgraph} of $H$ if $V(G) \subseteq V(H)$ and $E(G) \subseteq E(H)$.
\revise{
For a vertex set $U \subseteq V(H)$, the subgraph of $H$ induced by $U$, denoted $H[U]$, has vertex set $U$ and edge set $E(H[U]) = \big\{\{v, w\} \in E(H): v, w \in U\big\}$.
A subgraph $G$ of $H$ is an \emph{induced subgraph} of $H$ if  $G = H[V(G)]$.
}
Two graphs $H$ and $G$ are \emph{isomorphic} (denoted $H \cong G$), if there exists a bijective map ${\phi: V(H) \rightarrow V(G)}$ with $\{x, y\} \in E(H) \Leftrightarrow \{\phi(x)\phi(y)\} \in E(G)$ for all  $x,y \in V(H)$. 
We denote ${N_G(v) = \{u \in V(G): u \sim v\}}$ as the open and ${N_G[v] = N_G(v) \cup \{v\}}$ as the closed neighborhood of a vertex $v$.
\lr{The \textit{degree} of $v$ is defined as $d_G(v) = \abs{N(v)}$.
A vertex $v$ is \emph{isolated} if $d_G(v) = 0$, and a \emph{pendant} if $d_G(v) = 1$.}
We call a vertex $w$ \textit{confined} by a vertex $v$, if $N[w] \subseteq N[v]$ and \emph{universal}, if it is \revise{adjacent to}
\revisenote{changed it (from "connected") throughout the paper.} every other vertex in the graph.
    A \emph{path} of length $l\geq0$ is a sequence $(v_0, \cdots, v_{l})$ with $v_{i-1} \sim v_{i}$ for every $i \in [l]$.
    Unless stated otherwise, paths are consider to be  \emph{simple}, i.e., all vertices are pairwise distinct.
    A \emph{cycle} of length $l \geq 3$, is a path where only $v_0 = v_{l}$.
     The \emph{distance} $d_G(v,w)$ \lr{between two vertices $v$ and $w$ is the length (number of edges) of the shortest path between $v$ and $w$.}
    \lr{We omit subscripts when the graph is clear from the context.}

\paragraph{Graph Classes.}
A \textit{planar graph} is a graph $G$ that can be drawn in the plane so that edges only intersect in their endpoints.
A \emph{plane} graph is a planar graph together with a fixed planar embedding.
A \emph{chord} of a cycle $C$ is an edge between two non-consecutive vertices of $C$.
A graph $G$ is \emph{chordal} if it contains no chordless cycle of length at least four.
A graph $G$ is a \emph{split graph} if $V(G)$ can be partitioned into an independent set and a clique.
For any integer $t\geq 1$, a graph $G$ is $t$-claw-free (for $t=3$ simply claw-free), if $G$ has no induced star with $t$ leaves.

\paragraph{Domination.}
The \emph{domination number} is the minimum cardinality of a dominating set in $G$ and denoted as $\gamma(G)$.
A \emph{total dominating set} $S$ is a dominating set such that the induced subgraph of $G[S]$ has no isolated vertex.
A \emph{semitotal dominating set}~$S$ is a dominating set such that every vertex $d\in S$ has a \emph{witness} $d'\in S \setminus \{d\}$ such that  $d_G(d,d')\leq 2$. 
\cref{fig:dsexamples} illustrates a dominating set and a semitotal dominating set of a path graph. 

\paragraph{Parameterized Complexity. }
We only give the basic definitions and refer the interested reader to \cite{Cygan2015,Downey1999a}.
A parameterized problem is 
fixed-parameter tractable (FPT) if there exists an algorithm that correctly decides the problem in time ${f(k) \cdot {|G|}^{\cO(1)}}$ where $|G|$ is the size of the input, and $f$ is a computable function.
A \emph{kernel} is a polynomial-time algorithm that given an instance $\cI = (G, k)$ of a parameterized problem, outputs an equivalent instance $\cI' = (G', k')$ where $k' \leq g(k)$ and $|G'| \leq f(k)$, where $f$ and $g$ are computable functions.
Of particular interest are \emph{polynomial} (resp.\ \emph{linear}) kernels\revise{,} where $f(k)$ is a polynomial (resp.\ a linear function).
Kernels are often obtained by exhaustive application of reduction rules. 
A \emph{reduction rule} is a function $\phi$ that maps a parameterized instance $\cI = (G,k)$ to an equivalent instance $(G',k')$ such that $\phi$ is computable in time polynomial in $\abs{G}$ and $k$.
It is \emph{sound} if $(G, k) \in Q \Leftrightarrow (G',k') \in Q$.
We say a graph $G$ is \emph{reduced under} a set of (reduction) rules if neither of them can be applied to $G$ or the application creates $G'$ isomorphic to $G$.

The concept analogous to NP-hardness in parameterized complexity is that of 
$W[1]$-hardness, whose formal definition we omit. The basic assumption is that $FPT\neq W[1]$, under which no $W[1]$-hard problem admits an FPT-algorithm. In the $W$-hierarchy, $W[1]\subseteq W[2]$. 
A canonical complete problem for $W[1]$ is the parameterized version of \name{Independent Set} and for $W[2]$ is \dom.

\section{A Linear Kernel for Planar Graphs}\label{sec:planar}
\noindent We introduce reduction rules inspired by ideas in \cite{Garnero2018} for \tdom on planar graphs with the additional challenge of preserving the \lr{distance-$2$ witness constraint} of a semitotal dominating set.
We need the following definitions, similarly used in \cite{Garnero2018,Alber2004,Garnero2017}.

\begin{definition}\label{def:neighborhood-single}
For $v\in V$ we split $N(v)$ into three disjoint subsets.
Namely, we \revise{define the type-$1$, type-$2$ and type-$3$ neighorhoods as}
\begin{align}
        N_1(v) & := \{u \in N(v) : N(u) \setminus N[v] \neq \emptyset \}\,,              \\
        N_2(v) & := \{u \in N(v)\setminus N_1(v) : N(u) \cap N_1(v) \neq \emptyset \}\,, \text{and}\\
        N_3(v) & := N(v) \setminus \bigl(N_1(v) \cup N_2(v)\bigr)\,.
\end{align}

For two distinct vertices $v,w \in V$, we define ${N(v,w) := \big(N(v) \cup N(w)}\big)\setminus \{v,w\}$ to be the \revise{open} joint neighborhood.
We define $N_i[v] = N_i(v) \cup \{v\}$ as usual.
\begin{align}
    N_1(v,w) & := \{u \in N(v,w) \mid N(u) \setminus (N(v) \cup N(w) \cup \{v,w\}) \neq \emptyset \}\,,  \\
    N_2(v,w) & := \{u \in N(v,w)\setminus N_1(v,w) \mid N(u) \cap N_1(v,w) \neq \emptyset \}\,, \text{and} \\
    N_3(v,w) & :=  N(v,w) \setminus (N_1(v,w) \cup N_2(v,w))\,.
\end{align}
For $i \neq j$, we \revise{further simplify} notation write $N_{i,j}(v,w) := N_i(v,w) \cup N_j(v,w)$ and $N_{i,j} (v) := N_i(v) \cup N_j(v)$. 
\revise{We explicitly include $v$ and $w$ by writing $N_{i,j}[v,w]= N_{i,j}(v,w)\cup \{v, w\}$}.
\end{definition}
\lr{\noindent For a vertex $v\in V$, we say $N_{2,3}(v)$ \revise{is} the \emph{confined} part of neighborhood because for each $v'\in N_{2,3}(v)$, we have $N[v'] \subseteq N[v]$.}
To simplify the following arguments, we will use the notion of the \emph{witness radius}.
\begin{definition}[Witness radius]\label{def:witness-radius}
For \revise{a fixed graph $G$}, we define the \emph{witness radius} of a vertex $v$ as the set
\begin{align}
\WitRad(v) := \{v' \in V\setminus\{v\} \mid  d(v, v') \leq 2\}\,.
\end{align}
\end{definition}

\subsection{Reduction Rule One}
\begin{figure}[t]
\centering
    \resizebox{0.44\textwidth}{!}{
       \begin{tikzpicture}
	\begin{pgfonlayer}{nodelayer}
		\node [style=none] (58) at (-2.75, 1.5) {};
		\node [style=none] (59) at (-1.75, 1) {};
		\node [style=none] (60) at (-1.5, 0.5) {};
		\node [style=none] (61) at (-1.25, -0.75) {};
		\node [style=none] (62) at (-1.25, -1.5) {};
		\node [style=none] (63) at (-2, -2.75) {};
		\node [style=none] (64) at (-4.75, -2.75) {};
		\node [style=none] (65) at (-5.25, -2) {};
		\node [style=none] (66) at (-4.75, -1.5) {};
		\node [style=none] (67) at (-4.5, -1) {};
		\node [style=none] (68) at (-3.5, -1) {};
		\node [style=none] (69) at (-3, -0.25) {};
		\node [style=none] (70) at (-5.75, 0.25) {};
		\node [style=none] (71) at (-4.5, 0.25) {};
		\node [style=none] (72) at (-5.25, -1.5) {};
		\node [style=none] (73) at (-5.75, -1.5) {};
		\node [style=none] (52) at (-3, 2.25) {};
		\node [style=none] (53) at (-3.5, 1.5) {};
		\node [style=none] (54) at (-5.25, 0.5) {};
		\node [style=none] (55) at (-5.75, 0.5) {};
		\node [style=none] (56) at (-6, 1.25) {};
		\node [style=none] (57) at (-5, 2.5) {};
		\node [style=BLACK] (3) at (-3.5, 0) {};
		\node [style=NONE] (5) at (-3.5, 2) {};
		\node [style=NONE] (6) at (-4.75, 2) {};
		\node [style=NONE] (7) at (-5.5, 1) {};
		\node [style=NTWO] (8) at (-5, 0) {};
		\node [style=NTWO] (9) at (-5.5, -1) {};
		\node [style=NTHR] (10) at (-4.75, -2) {};
		\node [style=NTHR] (11) at (-3.5, -2) {};
		\node [style=NTHR] (12) at (-2.25, -2) {};
		\node [style=NTHR] (13) at (-2.5, 1) {};
		\node [style=NTHR] (14) at (-1.75, 0) {};
		\node [style=NTHR] (15) at (-1.75, -1) {};
		\node [style=none] (16) at (-3, 2.75) {};
		\node [style=none] (17) at (-3.5, 3) {};
		\node [style=none] (18) at (-4, 2.75) {};
		\node [style=none] (19) at (-4.75, 2.75) {};
		\node [style=none] (20) at (-5.5, 2.25) {};
		\node [style=none] (21) at (-5.25, 2.75) {};
		\node [style=none] (22) at (-6.5, 1) {};
		\node [style=none] (23) at (-6.25, 1.5) {};
		\node [style=none] (24) at (-6.25, 0.5) {};
		\node [style=DOM] (25) at (3.5, 0) {};
		\node [style=NONE] (26) at (3.5, 2) {};
		\node [style=NONE] (27) at (2.25, 2) {};
		\node [style=NONE] (28) at (1.5, 1) {};
		\node [style=BRGRAY] (29) at (2.25, 0) {};
		\node [style=BRGRAY] (30) at (1.5, -1) {};
		\node [style=BRGRAY] (31) at (2.25, -2) {};
		\node [style=BRGRAY] (32) at (3.5, -2) {};
		\node [style=BRGRAY] (33) at (4.75, -2) {};
		\node [style=NTHR, label={above:$p$}] (34) at (4.5, 1) {};
		\node [style=BRGRAY] (35) at (5.25, 0) {};
		\node [style=BRGRAY] (36) at (5.25, -1) {};
		\node [style=none] (37) at (4, 2.75) {};
		\node [style=none] (38) at (3.5, 3) {};
		\node [style=none] (39) at (3, 2.75) {};
		\node [style=none] (40) at (2.25, 2.75) {};
		\node [style=none] (41) at (1.5, 2.25) {};
		\node [style=none] (42) at (1.75, 2.75) {};
		\node [style=none] (43) at (0.5, 1) {};
		\node [style=none] (44) at (0.75, 1.5) {};
		\node [style=none] (45) at (0.75, 0.5) {};
		\node [style=none] (46) at (-3.25, 0.75) {$v$};
		\node [style=none] (48) at (3.75, 0.75) {$v$};
		\node [style=none] (49) at (-6, 3.5) {$G$};
		\node [style=none] (51) at (1, 3.5) {$G'$};
		\node [style=none] (74) at (-4.5, 3.25) {};
		\node [style=none] (75) at (-4.5, 3.25) {{\scriptsize $N_1(v)$}};
		\node [style=none] (76) at (-1, 1.5) {{\scriptsize $N_3(v)$}};
		\node [style=none] (77) at (-7, -0.75) {{\scriptsize $N_2(v)$}};
		\node [style=none] (78) at (-0.5, 0) {};
		\node [style=none] (79) at (0.25, 0) {};
	\end{pgfonlayer}
	\begin{pgfonlayer}{edgelayer}
		\draw [style=ULTRAGRAY] (58.center)
			 to [in=120, out=30] (59.center)
			 to (60.center)
			 to [in=90, out=-75, looseness=1.75] (61.center)
			 to (62.center)
			 to [in=0, out=-105] (63.center)
			 to [in=0, out=-180, looseness=0.25] (64.center)
			 to [in=-105, out=-180, looseness=1.50] (65.center)
			 to [in=-135, out=60] (66.center)
			 to (67.center)
			 to [in=-135, out=0, looseness=0.75] (68.center)
			 to (69.center)
			 to [in=210, out=60, looseness=0.75] cycle;
		\draw [style=ULTRAGRAY] (71.center)
			 to [bend right] (70.center)
			 to [in=180, out=-135, looseness=0.50] (73.center)
			 to [bend right=15] (72.center)
			 to [in=-60, out=45, looseness=0.50] cycle;
		\draw [style=ULTRAGRAY] (57.center)
			 to [in=120, out=45, looseness=0.50] (52.center)
			 to [in=30, out=-90, looseness=2.00] (53.center)
			 to (54.center)
			 to (55.center)
			 to [in=-105, out=105, looseness=1.25] (56.center)
			 to [in=225, out=75] cycle;
		\draw [style=EDGE] (3) to (7);
		\draw [style=EDGE] (3) to (6);
		\draw [style=EDGE] (3) to (5);
		\draw [style=EDGE] (6) to (5);
		\draw [style=EDGE] (6) to (7);
		\draw [style=EDGE] (8) to (7);
		\draw [style=EDGE] (8) to (3);
		\draw [style=EDGE] (8) to (9);
		\draw [style=EDGE] (10) to (3);
		\draw [style=EDGE] (11) to (3);
		\draw [style=EDGE] (12) to (3);
		\draw [style=EDGE] (10) to (9);
		\draw [style=EDGE] (10) to (11);
		\draw [style=EDGE] (11) to (12);
		\draw [style=EDGE, bend right] (10) to (12);
		\draw [style=EDGE] (13) to (3);
		\draw [style=EDGE] (14) to (3);
		\draw [style=EDGE] (15) to (3);
		\draw [style=EDGE] (14) to (15);
		\draw [style=EDGE] (19.center) to (6);
		\draw [style=EDGE] (6) to (20.center);
		\draw [style=EDGE] (6) to (21.center);
		\draw [style=EDGE] (5) to (18.center);
		\draw [style=EDGE] (16.center) to (5);
		\draw [style=EDGE] (17.center) to (5);
		\draw [style=EDGE] (24.center) to (7);
		\draw [style=EDGE] (7) to (22.center);
		\draw [style=EDGE] (23.center) to (7);
		\draw [style=EDGE] (25) to (28);
		\draw [style=EDGE] (25) to (27);
		\draw [style=EDGE] (25) to (26);
		\draw [style=EDGE] (27) to (26);
		\draw [style=EDGE] (27) to (28);
		\draw [style=ULTRAGRAY] (29) to (28);
		\draw [style=ULTRAGRAY] (29) to (25);
		\draw [style=ULTRAGRAY] (29) to (30);
		\draw [style=ULTRAGRAY] (31) to (25);
		\draw [style=ULTRAGRAY] (32) to (25);
		\draw [style=ULTRAGRAY] (33) to (25);
		\draw [style=ULTRAGRAY] (31) to (30);
		\draw [style=ULTRAGRAY] (31) to (32);
		\draw [style=ULTRAGRAY] (32) to (33);
		\draw [style=ULTRAGRAY, bend right] (31) to (33);
		\draw [style=EDGE] (34) to (25);
		\draw [style=ULTRAGRAY] (35) to (25);
		\draw [style=ULTRAGRAY] (36) to (25);
		\draw [style=ULTRAGRAY] (35) to (36);
		\draw [style=EDGE] (40.center) to (27);
		\draw [style=EDGE] (27) to (41.center);
		\draw [style=EDGE] (27) to (42.center);
		\draw [style=EDGE] (26) to (39.center);
		\draw [style=EDGE] (37.center) to (26);
		\draw [style=EDGE] (38.center) to (26);
		\draw [style=EDGE] (45.center) to (28);
		\draw [style=EDGE] (28) to (43.center);
		\draw [style=EDGE] (44.center) to (28);
		\draw [style=ULTRAGRAY] (28) to (30);
		\draw [style=EDGE] (7) to (9);
		\draw [style=simple directed edge] (78.center) to (79.center);
	\end{pgfonlayer}
\end{tikzpicture}
       }
    \caption{ {\bf a)} The gray line in the left figure shows how to split the neighborhood of a single vertex according to \cref{def:neighborhood-single}. {~\bf b)} An application of Rule~1: {Because $\abs{N_3(v)} \geq 1$, we replace $N_{2,3}(v):=N_2(v) \cup N_3(v)$ by a single pendant $p$.
    \revise{The witness radius $W_G(v)$ contains all shown vertices (without $v$) together with all neighbors of the $N_1(v)$ vertices.}
    }
    }
    \label{fig:ruleOne}
\end{figure}
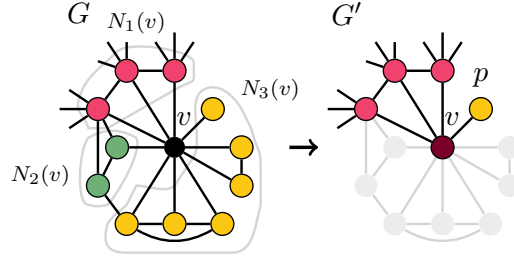

\lr{
\noindent The first rule is motivated by the following observation:
Any vertex $v' \in N_{2,3}(v)$ is confined in the neighborhood of $v$, and they will therefore never be useful to dominate other vertices.
Furthermore, it is easy to see that those vertices might only be important as a witness directly for $v$ itself, because every vertex that is of distance at most two from $v'$ is also of distance at most two from $v$.  
\begin{lemma}\label{lem:better-witness}
Fix arbitrary $v \in V$.
For any vertex $v' \in N_{2,3}(v)$, we have $\WitRad(v') \subseteq \WitRad(v)$.
\end{lemma}
\begin{proof}
     $v'$ is confined to the neighborhood of $v$: 
     the only way a path could leave $N[v]$ is through a vertex of $x \in N_1(v)$.
    But by definition, $x$ is a direct neighbor of $v$, so everything that is of distance at most two from $v'$ is also at \revise{a} distance of at most two from $v$.
    Thus, the \revise{lemma} follows.
\end{proof}
Hence, we can replace the whole $N_{2,3}(v)$-neighborhood by a single vertex that both forces $v$ into a dominating set and preserves a witness for $v$.
See a visualization in \cref{fig:ruleOne}.
}
This rule has already been introduced and proven sound for other domination variants; see, e.g., \cite{Alber2004,Garnero2018}. 
Recently, it was shown that it can be exhaustively applied in linear time $\cO(n + m)$ \cite{geis2025simplerbetterfasterstronger}.

\begin{rgl}\label{rgl:rone}
    Let \G be a graph and let $v \in V$. 
    
    If $\Nthreev \neq \emptyset$, then replace $N_{2,3}(v)$ with a single vertex $p$ \revise{adjacent} \revise{to} $v$.
\end{rgl}

\begin{lemma}[Soundness]\label{lemma:correctnessone}
    Let $G'$ be the graph obtained from $G$ by a \lr{single application} of Rule~1. 
    Then $\gamma_{t2}(G) = \gamma_{t2}(G')$.
\end{lemma}
\begin{proof}

\revise{Assume the rule has been applied to some vertex $v$} and let $p$ be the \revise{pendant of $v$ that replaces $N_{2,3}(v)$ in $G'$.} 
Recall from \cref{lem:better-witness} that for every $u\in N_{2,3}(v)$ we have $W_G(u)\subseteq W_G(v)$.

Let $D$ be a minimum semitotal dominating set of $G$.
Since $N_3(v)\neq\emptyset$, $D$ must contain some vertex from $N_{2,3}[v]$ to dominate $N_3(v)$.
If $v\notin D$, pick \revise{any} $u\in D\cap N_{2,3}(v)$ and replace $u$ by $v$; domination is preserved since
$N[u]\subseteq N[v]$, and the semitotal condition is preserved by $W_G(u)\subseteq W_G(v)$.
Thus assume WLOG that $v\in D$.
Now it is safe to delete all vertices of $D\cap N_{2,3}(v)$; this cannot destroy domination of vertices outside $N_{2,3}(v)$ as they are being dominated by $v$.
The only possible issue is that $v$ might lose its witness. 
If $v$ still has a witness in $D\setminus N_{2,3}(v)$, we are done. Otherwise $v$ was witnessed by some
$u\in D\cap N_{2,3}(v)$; replace that single $u$ by $p$. 
Then $p$ witnesses $v$ in $G'$ (and vice versa) while preserving the solution size.
Hence $\gamma_{t2}(G')\le \gamma_{t2}(G)$.

Let $D'$ be a minimum semitotal dominating set of $G'$.
We may assume $v\in D'$ \revise{because $N[p] \subseteq N[v]$ and $p$ must be dominated as well.}
Note that $p$ might be needed as a witness for $v$.
If $p\notin D'$, then $D'$ is immediately also a semitotal dominating set in $G$: when going from $G'$ back to $G$, witness distances between old vertices are preserved and the reintroduced vertices $N_{2,3}(v)$ are dominated by $v$.
If $p \in D'$, pick any $x\in N_3(v)$ (exists by assumption) and set $D:=\big(D'\setminus\{p\}\big)\cup\{x\}$.
Then $v$ and $x$ witness each other, domination is preserved, and $|D|=|D'|$.
Thus $\gamma_{t2}(G)\leq \gamma_{t2}(G')$.
Combining both directions gives $\gamma_{t2}(G)=\gamma_{t2}(G')$.
\end{proof}

\subsection{Reduction Rule Two}
\noindent Another observation is that, in the worst case, a solution might require up to four vertices in each joint \revise{open} neighborhood $N(v, w) := \big(N(v) \cup N(w)\big) \setminus \{v, w\}$ entirely: two to dominate and another two as witnesses.
\revise{Informally, one can think} of $N(v,w)$ as a small ``region'' that we want to simplify without changing $\gamma_{t2}$.
The key intuition is that vertices in $N_{2,3}(v,w)$ are only relevant insofar as they help to dominate the vertices in $N_3(v,w)$.
If there is a set of at most three vertices from $N_{2,3}(v,w)$ that already dominates $N_3(v,w)$, then there may be an optimal solution that avoids $v$ and $w$.
Conversely, if \emph{no} such small set exists, then any semitotal dominating set either uses $v$ or $w$, or it has to spend four vertices inside the region.
In the last case, it is always safe to \revise{swap} them with $v$ and $w$ for domination plus up to two additional vertices as witnesses.

\begin{definition}[Garnero and Sau {\cite{Garnero2018}}]\label{def:dvv}
Let $G$ be a graph and let $v \neq w$ be two vertices of $G$.
\revise{We now define the following three families of subsets:\revisenote{As marked by both reviewers: Removed the very vague "all subsets" part. }}
\begin{align}
    \Dvw & := \{ \tilde D \subseteq N_{2,3}(v,w)            \mid N_3(v,w) \subseteq \cup_{v' \in \tilde D} N(v'),\ |\tilde D| \leq 3                  \}\,, \\
    \Dv  & := \{ \tilde D \subseteq N_{2,3}(v,w) \cup \{v\} \mid N_3(v,w) \subseteq \cup_{v' \in \tilde D} N(v'),\ |\tilde D| \leq 3,\ v \in \tilde D \}\,,\\
    \Dw  & := \{ \tilde D \subseteq N_{2,3}(v,w) \cup \{w\} \mid N_3(v,w) \subseteq \cup_{v' \in \tilde D} N(v'),\ |\tilde D| \leq 3,\ w \in \tilde D \}\,.
\end{align}
\end{definition}

\noindent Intuitively, we first try \revise{all} subsets $\tilde D \subseteq N_{2,3}(v,w)$ of size at most three that dominate $N_3(v,w)$ \emph{without} using $v$ or $w$. 
If $\Dvw \neq \emptyset$, then an optimal solution \revise{possibly avoids both} $v$ and $w$. 
Conversely, if no such set exists, then at least one of $v$ and $w$ must be part of a solution.
\revise{Possible solutions involving these vertices are part of $\Dv$ and $\Dw$.}
\revise{In any case, we can dominate $N(v,w)$ by adding both $v$ and $w$ together with two neighbors as witnesses for $v$ and $w$ to a solution.}
\revisenote{Made the intution clearer and fixed the mistake of claiming to partition all subsets that cover $N(v, w)$.
}

\begin{observation}\label{fact:f1}
    Fix distinct $v,w \in V$.
If $\Dvw = \emptyset$ (with respect to $v$ and $w$), then $\gamma_{t2}(G) \leq k$  iff $G$ has a solution of size at most $k$ containing at least one of the two vertices $\{v,w \}$.
\end{observation}
\begin{proof}
Because $\Dvw = \emptyset$, a semitotal dominating set must contain at least one of $\{v, w \}$ or four vertices from $N_{2,3}(v,w)$: two for domination and two as witnesses.
In the second case, these four vertices can be replaced with $v$, $w$ and two neighbors of $v$ and $w$ as witnesses.
\end{proof}

\begin{observation}\label{fact:f2}
    Fix distinct $v,w \in V$.
    If $\Dvw = \emptyset$ and $\Dw = \emptyset$ (and resp. $\Dv = \emptyset$) then $\gamma_{t2}(G) \leq k$ iff $G$ has a semitotal dominating set of size at most $k$ containing $v$ (resp. $w$).
\end{observation}
\begin{proof}
As $\Dv = \emptyset$, no set of the form $\{v\}$, $\{v, u\}$ or $\{v, u, u'\}$ with ${u, u' \in N_{2,3}(v,w)}$ can dominate $N_3(v,w)$. 
Since also $\Dvw = \emptyset$, any semitotal dominating set in $G$ must contain $v$ or at least four other vertices by \cref{fact:f1}.
In the latter case, we replace these four vertices with $v$, $w$ and two neighbors as witnesses.
\end{proof}

\noindent \cref{fact:f1,fact:f2} allow us to introduce \emph{Rule~2}:
If $\Dv = \Dw = \emptyset$ (Case 1), we know by \cref{fact:f2} that $v$ and $w$ must be part of a solution.
Hence, we replace \revise{$N_{3}(v, w)$ with} two vertices $v'$ and $w'$ and \revise{add edge} $v \sim v'$ and $w \sim w'$.
\revise{Note that by this construction, the vertices in $N_2(v,w)$ are immediately dominated by $v$ and $w$ as well but we can not remove them, because they would destroy witness relations between $v$, $w$ and the vertices in $N_1(v,w)$.
}
\revise{Unfortunately, removing $N_{3}(v,w)$ could increase $d(v,w)$ and we might loose the possibility of $v$ and $w$ witnessing each other. 
We therefore need to artificially preserve the distance.
}
If $\Dv \neq \emptyset$ but $\Dw = \emptyset$, by \cref{fact:f2} we know that $v$ must be part of a solution, and we simply replace $N_{2,3}(v)$ with a single pendant $v'$ preserving a witness for $v$ as well.
Find an example in \cref{fig:ruleTwo}.

\begin{figure}[t]
    \centering
    \resizebox{0.90\textwidth}{!}{
        \input{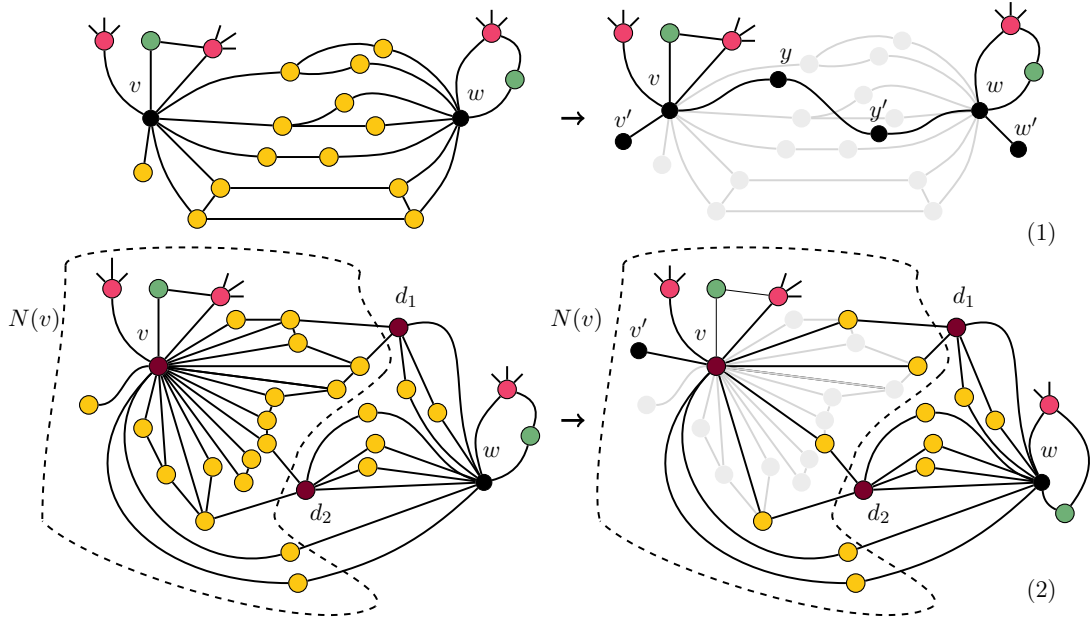}
    }
    \caption[Application of \revise{Rule~2}]{\textit{Application of Rule~2: (1) $\Dv = \Dw = \emptyset$ and Case 1 applies.
    (2) Case \ref{case:c2} has been applied and \revise{$N_{3}(v)$} removed. 
    Observe that those vertices that cross the dotted line are vertices from $N_1(v)$ and not removed. \revise{**Changed: Also adjusted the figure to reflect the preservation of type-2 neighbors.**}
    }}
    \label{fig:ruleTwo}
\end{figure}

\begin{rgl}\label{rgl:rtwo}
    Fix distinct $v, w$.
    If $\mathbf{\Dvw = \emptyset}$ with respect to $v$ and $w$, we apply the following:
    \begin{enumerate}
        \item if $\Dv =  \emptyset$ and $\Dw = \emptyset$, then
            \begin{itemize}
                \item Add \lr{pendant} vertices $v'$ and $w'$ with $v \sim v'$ and $w \sim w'$ and remove \revise{$N_{3}(v,w)$},
                \item If \revise{$N(v) \cap N(w) \cap N_3(v, w) \neq \emptyset$}, then ,add another vertex $y$ and two \revise{adjacent} edges  $v \sim y \sim w$, and
                \item If \revise{$N_3(v) \cap N(w) \cap N_3(v,w) = \emptyset$}, but at least one path of length three from $v$ to $w$, add two vertices $y$ and $y'$ with $v \sim y \sim y' \sim w$\,.
             \end{itemize}
       \item if $\Dv \neq  \emptyset$ and $D_w = \emptyset$, then  remove \revise{$N_{3}(v)$}, and add a new pendant $v'\sim v$. \label{case:c2}
       \item if $\Dv =  \emptyset$ and $D_w \neq \emptyset$ (symmetric).\label{case:c3}
        \end{enumerate}
\end{rgl}
\noindent \revise{The remaining case $\Dv \neq \emptyset$ and $\Dw \neq \emptyset$ is not required for getting a bound on the kernel size.}

\begin{lemma}[Soundness of \revise{Rule~2}]\label{lemma:correctnesstwo}
    Let $G = (V,E)$ be a graph,  $v, w \in V$ and $G'$ be the graph after applying \revise{Rule~2} on $v$ and $w$ \lr{once}.
    Then $\gamma_{t2}(G) = \gamma_{t2}(G')$.
\end{lemma}
\begin{proof}
Assume a semitotal dominating set $D$ in $G$ and $\Dvw = \emptyset$ with respect to the fixed pair $v$ and $w$, hence \revise{Rule~2} applies.
Note that we assumed that $G$ has no isolated vertices.
A semitotal dominating set $D'$ in $G'$ can now be obtained by first setting $D' \coloneq D$ and then perform small modifications on the vertices in \revise{$N_{3}(v, w) \cap D$}, depending on the case.

Assume the first case, i.e. \revise{$ \Dv = \emptyset$ and $\Dw = \emptyset $}. 
By applying \cref{fact:f2} twice, \revise{both} $\{v,w\} \subseteq D$ and no vertex in $N_{2,3}(v,w)$ is necessary to \revise{solely} dominate other vertices.
If $d(v,w) \leq 2$ then $v$ and $w$ immediately witness each other, so we can assume \revise{$N_{2,3}(v, w) \cap D = \emptyset$} (otherwise, removing them yields a better solution).
On the contrary, if $d(v,w ) = 3$ then a single vertex $d \in N_{2,3}(v,w) \cap D$ might exist as a witness for both $v$ and $w$.
In that case, the rule preserves this by keeping the path $v\sim y \sim y' \sim w$ and we simply replace $d$ with $y$ in $D'$ and we can safely remove $N_{2,3}(v, w)$.
If $d(v, w) > 3$,  then the neighborhoods of $v$ and $w$ are disjoint and they could have had two distinct witnesses in $N(v, w)$.
In that case, we replace them with $v'$ and $w'$ in $D'$.
\revise{Finally, observe that $v$ and $w$ can be  witnesses for the vertices in $N_1(v, w)$ via type-$2$ neighbors. 
Because our reduction only removes type-$3$ neighbors, these relations are being preserved.}

Now assume $\Dv \neq \emptyset$ and $\Dw = \emptyset$ (the other case is symmetric).
By \cref{fact:f2}, we know $v \in D$ and $v$ dominates $N_{2,3}(v)$. 
Note that $\WitRad(v) \supseteq \WitRad(x)$ for every $x \in N_{2,3}(v)$, so $x$ was only used as a witness for $v$ directly in which case we replace it with the single pendant $v'$ in $D'$.

\revise{Let $D'$ be a semitotal dominating set in $G'$ and let us show how} to replace vertices in $D'$ to get a solution $D$ in $G$.
If $\Dv \neq \emptyset$ and $\Dw = \emptyset$ then we assume $v \in D$ and also $v'\in D'$ happens only if it is a witness for $v$ \big(trivially $\WitRad(v)\supset \WitRad(v')$\big).
In that case, we replace it with an arbitrary vertex in $N(v)$ in $G$ which exists by assuming no isolated vertices.
The symmetric case works the same.

Now assume $\Dv = \emptyset$ and $\Dw = \emptyset$.
First, assume $\{v, w\} \subseteq D$, enforced by $v'$ and $w'$.
The rule is designed to preserve the exact distance between $v$ and $w$ and the new vertices $y$ and potentially $y'$ are only used in a solution to preserve witnessing properties, and they can be trivially mapped back to corresponding vertices in $G$.
Thus we have shown in all cases $\gamma_{t2}(G) = \gamma_{t2}(G')$.
\end{proof}

\noindent \revise{Our next rule will use so-called \emph{simple regions.}
We defer a formal treatment to \cref{ch:kernel-size}, as the following only requires some intuition (also visualized in \cref{fig:rulethree}).
}

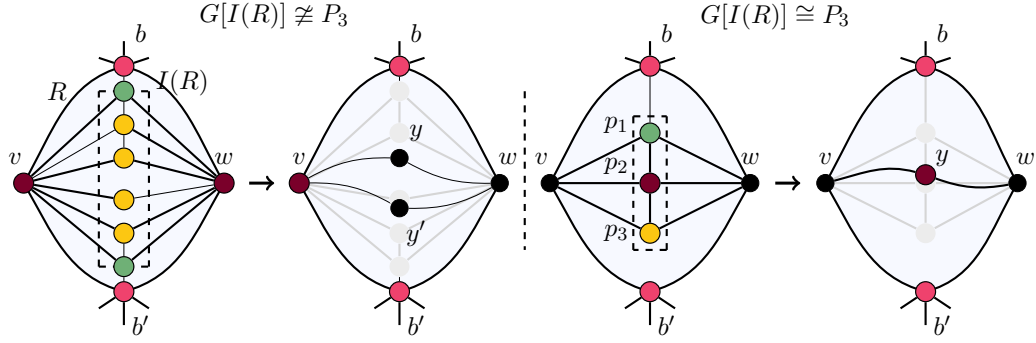
\begin{figure}[t]
    \centering
    \resizebox{0.85\textwidth}{!}{
        \begin{tikzpicture}
	\begin{pgfonlayer}{nodelayer}
		\node [style=DOM] (0) at (-15, 0) {};
		\node [style=DOM] (1) at (-9, 0) {};
		\node [style=NONE] (2) at (-12, 3.5) {};
		\node [style=NTWO] (3) at (-12, 2.75) {};
		\node [style=NTHR] (4) at (-12, 1.75) {};
		\node [style=NTHR] (5) at (-12, 0.75) {};
		\node [style=NTHR] (8) at (-12, -0.5) {};
		\node [style=NTHR] (9) at (-12, -1.5) {};
		\node [style=NTWO] (10) at (-12, -2.5) {};
		\node [style=NONE] (11) at (-12, -3.25) {};
		\node [style=DOM] (13) at (-6.75, 0) {};
		\node [style=BLACK] (14) at (-0.75, 0) {};
		\node [style=NONE] (15) at (-3.75, 3.5) {};
		\node [style=BRGRAY] (16) at (-3.75, 2.75) {};
		\node [style=BRGRAY] (17) at (-3.75, 1.5) {};
		\node [style=BRGRAY] (21) at (-3.75, -0.5) {};
		\node [style=BRGRAY] (22) at (-3.75, -1.5) {};
		\node [style=BRGRAY] (23) at (-3.75, -2.5) {};
		\node [style=NONE] (24) at (-3.75, -3.25) {};
		\node [style=BLACK] (25) at (-3.75, 0.75) {};
		\node [style=BLACK] (26) at (-3.75, -0.75) {};
		\node [style=none] (27) at (-15.25, 0.75) {$v$};
		\node [style=none] (28) at (-9, 0.75) {$w$};
		\node [style=none] (31) at (-0.5, 0.75) {$w$};
		\node [style=none] (32) at (-6.75, 0.75) {$v$};
		\node [style=none] (33) at (-3.25, 1.5) {$y$};
		\node [style=none] (34) at (-3.25, -1.5) {$y'$};
		\node [style=none] (35) at (-3, 3.75) {};
		\node [style=none] (36) at (-3.75, 4.25) {};
		\node [style=none] (37) at (-4.5, 3.75) {};
		\node [style=none] (38) at (-11.25, 3.75) {};
		\node [style=none] (39) at (-12, 4.25) {};
		\node [style=none] (40) at (-12.75, 3.75) {};
		\node [style=none] (41) at (-12.75, -3.75) {};
		\node [style=none] (42) at (-11.25, -3.75) {};
		\node [style=none] (43) at (-12, -4.25) {};
		\node [style=none] (44) at (-4.5, -3.75) {};
		\node [style=none] (45) at (-3.75, -4.25) {};
		\node [style=none] (46) at (-3, -3.75) {};
		\node [style=BLACK] (47) at (0.75, 0) {};
		\node [style=BLACK] (48) at (6.75, 0) {};
		\node [style=NONE] (49) at (3.75, 3.5) {};
		\node [style=NTWO] (50) at (3.75, 1.5) {};
		\node [style=DOM] (52) at (3.75, 0) {};
		\node [style=NTHR] (54) at (3.75, -1.5) {};
		\node [style=NONE] (56) at (3.75, -3.25) {};
		\node [style=BLACK] (58) at (9, 0) {};
		\node [style=BLACK] (59) at (14.75, 0) {};
		\node [style=NONE] (60) at (12, 3.5) {};
		\node [style=BRGRAY] (61) at (12, 1.5) {};
		\node [style=BRGRAY] (63) at (12, 0) {};
		\node [style=BRGRAY] (64) at (12, -1.5) {};
		\node [style=NONE] (66) at (12, -3.25) {};
		\node [style=DOM] (67) at (12, 0.25) {};
		\node [style=none] (69) at (0.5, 0.75) {$v$};
		\node [style=none] (70) at (6.75, 0.75) {$w$};
		\node [style=none] (71) at (15, 0.75) {$w$};
		\node [style=none] (72) at (9, 0.75) {$v$};
		\node [style=none] (74) at (12.5, 0.75) {$y$};
		\node [style=none] (75) at (12.75, 3.75) {};
		\node [style=none] (76) at (12, 4.25) {};
		\node [style=none] (77) at (11.25, 3.75) {};
		\node [style=none] (78) at (4.5, 3.75) {};
		\node [style=none] (79) at (3.75, 4.25) {};
		\node [style=none] (80) at (3, 3.75) {};
		\node [style=none] (81) at (3, -3.75) {};
		\node [style=none] (82) at (4.5, -3.75) {};
		\node [style=none] (83) at (3.75, -4.25) {};
		\node [style=none] (84) at (11.25, -3.75) {};
		\node [style=none] (85) at (12, -4.25) {};
		\node [style=none] (86) at (12.75, -3.75) {};
		\node [style=none] (87) at (0, 2.75) {};
		\node [style=none] (88) at (0, -2) {};
		\node [style=none] (89) at (0, -2) {};
		\node [style=none] (90) at (7.5, 5) {$G[I(R)] \cong P_3$};
		\node [style=none] (91) at (-7.5, 5) {$G[I(R)] \ncong P_3$};
		\node [style=none] (92) at (-12.75, 2.75) {};
		\node [style=none] (93) at (-11.25, 2.75) {};
		\node [style=none] (94) at (-12.75, -2.5) {};
		\node [style=none] (95) at (-11.25, -2.5) {};
		\node [style=none] (96) at (-11.25, 2.75) {};
		\node [style=none] (97) at (3.25, 2) {};
		\node [style=none] (98) at (4.25, 2) {};
		\node [style=none] (99) at (4.25, -2) {};
		\node [style=none] (100) at (3.25, -2) {};
		\node [style=none] (101) at (-10.25, 3) {$I(R)$};
		\node [style=none] (108) at (2.75, 1.75) {$p_1$};
		\node [style=none] (109) at (2.75, 0.5) {$p_2$};
		\node [style=none] (110) at (2.75, -1.5) {$p_3$};
		\node [style=none] (111) at (-11.5, 4.5) {$b$};
		\node [style=none] (112) at (12.5, 4.5) {$b$};
		\node [style=none] (113) at (4.25, 4.5) {$b$};
		\node [style=none] (114) at (-3.25, 4.5) {$b$};
		\node [style=none] (115) at (4.25, -4.25) {$b'$};
		\node [style=none] (116) at (12.5, -4.25) {$b'$};
		\node [style=none] (117) at (-3.25, -4.25) {$b'$};
		\node [style=none] (118) at (-11.5, -4.25) {$b'$};
		\node [style=none] (119) at (-14, 2.75) {$R$};
		\node [style=none] (120) at (7.5, 0) {};
		\node [style=none] (121) at (8.25, 0) {};
		\node [style=none] (122) at (-8.25, 0) {};
		\node [style=none] (123) at (-7.5, 0) {};
	\end{pgfonlayer}
	\begin{pgfonlayer}{edgelayer}
		\draw [style=FILLBLUE] (11.center)
			 to [in=240, out=15] (1.center)
			 to [in=-15, out=120] (2.center)
			 to [in=60, out=-165] (0.center)
			 to [in=165, out=-60] cycle;
		\draw [style=EDGE] (3) to (0);
		\draw (0) to (4);
		\draw [style=EDGE] (5) to (1);
		\draw [style=EDGE] (0) to (5);
		\draw [style=EDGE] (4) to (1);
		\draw (1) to (8);
		\draw [style=EDGE] (8) to (0);
		\draw [style=EDGE] (0) to (9);
		\draw [style=EDGE] (9) to (1);
		\draw [style=EDGE] (10) to (0);
		\draw [style=EDGE] (1) to (10);
		\draw [style=EDGE] (3) to (1);
		\draw (10) to (11);
		\draw (3) to (2);
		\draw (10) to (9);
		\draw (3) to (4);
		\draw (5) to (4);
		\draw [style=FILLBLUE] (13.center)
			 to [in=165, out=-60] (24.center)
			 to [in=240, out=15] (14.center)
			 to [in=-15, out=120] (15.center)
			 to [in=60, out=-165] cycle;
		\draw [style=ULTRAGRAY] (16) to (13);
		\draw [style=ULTRAGRAY] (13) to (17);
		\draw [style=ULTRAGRAY] (17) to (14);
		\draw [style=ULTRAGRAY] (14) to (21);
		\draw [style=ULTRAGRAY] (21) to (13);
		\draw [style=ULTRAGRAY] (13) to (22);
		\draw [style=ULTRAGRAY] (22) to (14);
		\draw [style=ULTRAGRAY] (23) to (13);
		\draw [style=ULTRAGRAY] (14) to (23);
		\draw [style=ULTRAGRAY] (16) to (14);
		\draw (23) to (24);
		\draw (16) to (15);
		\draw [style=ULTRAGRAY] (23) to (22);
		\draw [style=ULTRAGRAY] (16) to (17);
		\draw [bend right=15] (25) to (14);
		\draw [bend left=15] (14) to (26);
		\draw [bend right=15] (26) to (13);
		\draw [bend left=15] (13) to (25);
		\draw [style=EDGE] (40.center) to (2);
		\draw [style=EDGE] (2) to (39.center);
		\draw [style=EDGE] (38.center) to (2);
		\draw [style=EDGE] (37.center) to (15);
		\draw [style=EDGE] (36.center) to (15);
		\draw [style=EDGE] (35.center) to (15);
		\draw [style=EDGE] (43.center) to (11);
		\draw [style=EDGE] (41.center) to (11);
		\draw [style=EDGE] (42.center) to (11);
		\draw [style=EDGE] (24) to (44.center);
		\draw [style=EDGE] (24) to (45.center);
		\draw [style=EDGE] (46.center) to (24);
		\draw [style=FILLBLUE] (56.center)
			 to [in=240, out=15] (48.center)
			 to [in=-15, out=120] (49.center)
			 to [in=60, out=-165] (47.center)
			 to [in=165, out=-60] cycle;
		\draw [style=EDGE] (50) to (47);
		\draw [style=EDGE] (52) to (48);
		\draw [style=EDGE] (47) to (52);
		\draw [style=EDGE] (47) to (54);
		\draw [style=EDGE] (54) to (48);
		\draw [style=EDGE] (50) to (48);
		\draw (50) to (49);
		\draw [style=FILLBLUE] (58.center)
			 to [in=165, out=-60] (66.center)
			 to [in=240, out=15] (59.center)
			 to [in=-15, out=120] (60.center)
			 to [in=60, out=-165] cycle;
		\draw [style=ULTRAGRAY] (61) to (58);
		\draw [style=ULTRAGRAY] (59) to (63);
		\draw [style=ULTRAGRAY] (63) to (58);
		\draw [style=ULTRAGRAY] (58) to (64);
		\draw [style=ULTRAGRAY] (64) to (59);
		\draw [style=ULTRAGRAY] (61) to (59);
		\draw [style=ULTRAGRAY] (61) to (60);
		\draw [style=EDGE, bend right=15] (67) to (59);
		\draw [style=EDGE, in=165, out=15] (58) to (67);
		\draw [style=EDGE] (80.center) to (49);
		\draw [style=EDGE] (49) to (79.center);
		\draw [style=EDGE] (78.center) to (49);
		\draw [style=EDGE] (77.center) to (60);
		\draw [style=EDGE] (76.center) to (60);
		\draw [style=EDGE] (75.center) to (60);
		\draw [style=EDGE] (83.center) to (56);
		\draw [style=EDGE] (81.center) to (56);
		\draw [style=EDGE] (82.center) to (56);
		\draw [style=EDGE] (66) to (84.center);
		\draw [style=EDGE] (66) to (85.center);
		\draw [style=EDGE] (86.center) to (66);
		\draw [style=EDGEDASHED] (87.center) to (89.center);
		\draw [style=EDGE] (52) to (54);
		\draw [style=ULTRAGRAY] (64) to (63);
		\draw [style=EDGE] (52) to (50);
		\draw [style=ULTRAGRAY] (63) to (61);
		\draw [style=EDGEDASHED] (92.center) to (96.center);
		\draw [style=EDGEDASHED] (96.center) to (95.center);
		\draw [style=EDGEDASHED] (94.center) to (92.center);
		\draw [style=EDGEDASHED] (94.center) to (95.center);
		\draw [style=EDGEDASHED] (97.center) to (98.center);
		\draw [style=EDGEDASHED] (98.center) to (99.center);
		\draw [style=EDGEDASHED] (99.center) to (100.center);
		\draw [style=EDGEDASHED] (97.center) to (100.center);
		\draw [style=simple directed edge] (120.center) to (121.center);
		\draw [style=simple directed edge] (122.center) to (123.center);
	\end{pgfonlayer}
\end{tikzpicture}
        }
    \caption[Application of \cref{rgl:rthree}]{
    Rule~3 on a simple region between $v$ and $w$ to the interior vertices $I(R) = V(R) \setminus \partial R$.
    
    Left: the vertices inside the region are not isomorphic to a $P_3$ triggering Case \ref{case:rtt}. Right: Case \ref{case:rto} and the whole inner region is replaced with one single vertex.}
    \label{fig:rulethree}
\end{figure}

\begin{definition}[Simple region (informal)]
A simple region $R$ between two vertices $v$ and $w$ in a plane embedding is described by
\begin{enumerate*}[label=\allbold{(\alph*)}]
    \item two distinct boundary paths $P_1$ and $P_2$ from $v$ to $w$ of length at most two, and 
   \item all vertices enclosed by $P_1$ and $P_2$ are common neighbors of both $v$ and $w$.
\end{enumerate*}
\end{definition}

We denote the set of vertices on the boundary paths $\partial(R) = V(P_1)\cup V(P_2)$ and the vertices in the interior as $I(R) = V(R) \setminus \partial(R)$.
\lr{Furthermore in the following, for a region $R$ with poles $v,w$, we denote with $N_{2,3}(R)$ \big(and analogous $N_2(R)$, $N_3(R)$\big) \revise{to} be the subset of $N_{2,3}(v,w)$ that is enclosed (or on) the boundary paths of $R$.}
A \textit{simple region} has at most two vertices from $N_1(v,w)$ (potentially the two vertices on the boundary path), at most two vertices from $N_2(v,w)$ \revise{adjacent} to the border, and \revise{an} unlimited \revise{number of vertices from $N_3(v,w)$} squeezed in the middle.
Note that there can be multiple simple regions between $v$ and $w$ at the same time in a fixed plane embedding.\revisenote{We deferred the formal definition of a simple region to a bit later to improve readability.}

\subsection{Reduction Rule Three}
\noindent The main idea is simple: whenever there are sufficiently many vertices in $I(R)$ then $v$ or $w$ can be assumed to be in a semitotal dominating set because they dominate all of $I(R)$ by definition.
In fact, we will show that $|I(R)| \geq 3$ is enough to apply the rule.
A vertex $d \in I(R)$ is only forced to be in a solution in two cases: either $d$ dominates all of $I(R)$ alone, or it is required as a witness for any other vertex in $\bigcup_{b \in \partial(R)} N[b]$.
The first case can only happen if $G[I(R)] \cong P_3$ (a path on three vertices).
Then we replace $I(R)$ with a single vertex $y$, adjacent to both $v$ and $w$.
Otherwise, it is safe to assume either $v$ (or symmetrically $w$) in a solution, and furthermore, a vertex in $I(R)$ \revise{is} only \revise{used} for witnessing $v$ directly or any neighbor of $N(w)$ (resp. $N(v)$) because the neighbors of $\partial(R)$ are already witnessed by $v$ or $w$.
We can preserve these properties by replacing the interior with two vertices $\{y,y'\}$ and two paths $v \sim y \sim w$ and $v \sim y' \sim w$.
This simulates an OR-gadget for choosing $v$ versus $w$ and simultaneously preserves the existence of required distance-2 witnesses.

\begin{rgl}\label{rgl:rthree}
    \revise{Fix any pair of distinct vertices $v, w \in V$ and let $R$} be \revise{any} simple region between $v$ and $w$ and 
    If $\abs{N_{2,3}(R)} \geq 3$, then remove $I(R)$, and do the following.
    
    \noindent \revise{\textit{Case 1:}\label{case:rto}} If $G[I(R)] \cong P_3$, then add a new vertex $y$ with edges $v \sim y$ and $y \sim w$, and

    \noindent \revise{\textit{Case 2:}\label{case:rtt}} otherwise,  add vertices $y$, $y'$ and edges $v \sim y \sim w$, and $v \sim y' \sim w$\,.
\end{rgl}

\begin{lemma}[Soundness of \revise{Rule~3}]\label{lemma:correctnessthree}
    Let $G$ be a plane graph, $v, w \in V$ and \GB be the graph obtained after one application of \revise{Rule~3} on the pair $\{v, w\}$. 
    Then $\gamma_{t2}(G)  = \gamma_{t2}(G')$.
\end{lemma}
\begin{proof}
        Let $R$ be a simple region between $v$ and $w$ with border vertices $\partial(R) = \{b_1, b_2\}$ and denote $I(R) = V(R) \setminus \partial(R)$ and recall $\WitRad$ as defined in \cref{def:witness-radius}.
        First, consider a semitotal dominating set $D$ in $G$ and we show that $G'$ also has a semitotal dominating set $D'$ with $\abs{D'} \leq \abs{D}$.
        
        First, assume that $G[I(R)] \ncong P_3$, then we can assume $\{v,w\} \cap D \neq \emptyset$ because either $\{b_1,b_2,s\} \subseteq D$, $\{b_1, s\}\subseteq D$ (resp. $b_2$) for a single vertex $s \in I(R)$ or simply two vertices $s_1, s_2 \in I(R)$ are in a solution.
        In the first case, we can replace $s$ with either $v$ or $w$ because \revise{$\WitRad(s) \subseteq \WitRad(b_1) \cup \WitRad(b_2)$}.\revisenote{
        Thanks to R1 for pointing out this mistake! The inclusion must be reverted.
        The argument in this case is that if both border vertices ar already part of a solution, no interior vertex is required as a witness alone.
        }
        In the second case, observe that $\WitRad(s) \subseteq \WitRad(v) \cup \WitRad(b_1)$ and the same replacement argument is sound, and in the last case, we can again simply replace them by $v$ and $w$.
        
        \revise{If} both $\{v,w\} \subseteq D$, then assume $I(R) \cap D = \emptyset$, because $\WitRad(v) \cup \WitRad(w) \supseteq \WitRad(d)$ for all $d \in I(R)$. 
        Now assume $\{v,s\} \subseteq D$ for some $s \in I(R)$. 
        Then, by a similar argument, we can replace it with $w$, as $\WitRad(v) \cup \WitRad(s) \subseteq \WitRad(v) \cup \WitRad(w)$, and $w$ dominates at least those vertices that $s$ did.
        In both cases, set $D' = D$.
        Now, assume that $G[I(R)] \cong P_3$ and denote the middle vertex as $p_2$. 
        In addition to the previous cases, $p_2$ can dominate the entire inner region $I(R)$ \emph{and} provide a witness for $N(v) \cup N(w) \cup \{b_1, b_2\}$ simultaneously, which can only be covered with $v$ and $w$ together.
        This time, if $p_2 \in D$ , we  replace it with $y$ in $G'$.

        On the contrary, assume some semitotal dominating set $D'$ in $G'$.
        First, assume Case $1$ has applied.
        If Case $1$ has applied, then we had $G[I(R)] \cong P_3$ with vertices consecutive $p_1,p_2,p_3$.
        If the single vertex $y \in D'$ then we replace it with $p_2$ in $G$ which dominates $p_1$ and $p_3$.
        Furthermore, $\WitRad(y) = \WitRad(p_2)$, so this operation is sound.
        If Case $2$ has applied, it is trivial to see that none of $y_1$ or $y_2$ is part of a solution, because in every case, we can simply \revise{swap} them to $v$ or $w$ while preserving the witnesses.
        In all cases, we have $\gamma_{t2}(G) = \gamma_{t2}(G')$.
\end{proof}

\begin{corollary}\label{running-time}
All three rules can be \emph{exhaustively} applied in time $\mathcal{O}(n^3)$ on a graph $G$ of size $n$.
\end{corollary}
\begin{proof}
\lr{
Rule~1 can exhaustively be applied in linear time - even on general graphs due to \cite{geis2025simplerbetterfasterstronger}.
As mentioned in~\cite{Garnero2018,Alber2002}, we can construct $\Dvw, \Dv$ and $\Dw$ in fpt-time ${\mathcal{O}(2^{\sqrt{3}}\cdot (d(v) + d(w))}$.
As in \cite{Garnero2018}, Rule~3 can be applied in time $\cO\big(d(v) + d(w)\big)$ dependent on the degrees.
Our check for an induced $P_3$ can already be incorporated in a single scan of the neighborhoods.
Finding a pair where \revise{Rule~2} or \revise{Rule~3} can \revise{be applied in time $\cO(n^2)$} by simply testing all of them.\revisenote{Check this: Constant time check?}
In each step, the graph gets reduced by at least one vertex, leading to at most $\mathcal{O}(n)$ rounds, giving a total running time of $\mathcal{O}(n^3)$.

Note that some rules \revise{Rule~1} can be applied multiple times without changing the graph.
To prevent this, we need to check whether the graph has been reduced or the \revise{returned graph $G'$ is isomorphic to $G$.}
Because of the locality of our reduction rules, we can simply keep track of this in constant time and will ignore it at this point.}

\end{proof}

\subsection{Bounding the Kernel Size}\label{ch:kernel-size}
\noindent To bound the kernel size, we first introduce a region-decomposition framework developed by Alber, Fellows and Niedermeier \cite{Alber2004} \lr{which the authors later generalized in \cite{Alber2006}} \footnote{We also refer to \cite{Garnero2018} for a rigorous formal treatment of their ideas.}
Fix a plane embedding of a planar graph $G$ and a (semitotal) dominating set $D$ of size $k$. 
There exists a maximal decomposition of the graph embedding into at most $\mathcal{O}(k)$ regions, such that regions \lr{intersect only on boundary edges} and the vertices from $D$ lie on these boundary paths.
Using Rules~2~and~3, we obtain a constant upper bound on the number of vertices contained in any region.
Though, some vertices might not be enclosed by any region.
As we will see later, Rule~1 allows us \revise{to shrink the size of a simple region to a $\cO(1)$ number of vertices.}

\begin{definition}[\citep{Alber2004} Region Decompositions]
    Let ${G = (V, E)}$ be a plane graph. A \emph{region} $R(v, w)$ between two \lr{distinct} vertices $v,w$ is a closed subset of the plane with the following properties:
    \begin{enumerate}
    \itemsep0pt
    \item the boundary of $R(v, w)$ is formed by two paths $P_1$ and $P_2$ of length at most three in $G$ that connect $v$ and $w$, and
    \item all vertices that are strictly inside the region $R(v, w)$ are from $N(v, w)$.
    \end{enumerate}
    We denote the vertices on the boundary with $\partial R$ and denote the vertices belonging to $R$ with 
    \begin{align}
    V(R) = \{u\in V \mid u \text{~sits inside or on the boundary of } R\}\,.
    \end{align}
    \noindent \textbf{$D$-region decomposition.} For a subset $D \subseteq V$, a $D$-\emph{region decomposition} of $G$ is a set $\fR$ of regions between pairs of vertices in $D$ such that
    \begin{enumerate}
    \itemsep0pt
        \item  for $R(v, w) \in \fR$ no vertex from $D$ (except $v,w$) lies in $V(R(v, w))$, and 
        \item for two regions $R_1, R_2 \in \fR$, it holds $(R_1 \cap R_2) \subseteq (\partial R_1 \cup \partial R_2)$.
    \end{enumerate}
    \noindent A \dreg is \underline{maximal} if there is no region $R \notin \fR$ such that $\fR' = \fR \cup \{R\}$ is a \dreg~with $V(\fR) \subsetneq V(\fR')$ \lr{where $V(\fR):= \cup_{Q \in \fR}V(Q)$}.
    We define the multigraph (allowing parallel edges) $G_{\fR}$ of a $D$-region decomposition of $G$ with $V_{\fR} = D$ and $E_{\fR} = \{v \sim  w\mid \exists \text{region~} R(v, w) \text{~between~} v,w \in D\}$.
\end{definition}

\noindent \lr{Furthermore, \revise{they showed that} $\abs{\fR} = \Theta(|D|)$ for a maximal $D$-region decomposition if $D$ is any dominating set.}
\begin{lemma}[{\cite[Prop. 1 and Lemma 5]{Alber2004}} Linear Decomposition]\label{lemma:numRegions}
    Let $G$ be a plane graph and let $D$ be a dominating set with $\abs{D}\geq 3$. 
    Then, there is a maximal $D$-region decomposition $\fR$ of $G$ such that $\abs{\fR} \leq 3 \cdot \abs{D} - 6$.
\end{lemma}
\noindent Because the set of all semitotal dominating sets is a subset of all dominating sets, \cref{lemma:numRegions} trivial applies to our problems as well.
Using these definitions, we can now precisely describe a \emph{simple region}:

\begin{definition}[Simple Region]
    A \emph{simple} $vw$-region in a plane embedding is a $vw$-region such that its boundary paths have length at most $2$, and  all vertices in $V(R)$ are common neighbors of both poles $v,w$, i.e.  ${V(R) \setminus \{v,w\} \subseteq N(v) \cap N(w)}$.
\end{definition}

\subsubsection{Bounding the Size of a Region}

\noindent We start by analyzing the impact of our reduction rules on a single $vw$-region $R$. 
If $\revise{\Dvw} = \emptyset$ (for the poles $v$ and $w$), then Rule 2 must have been applied, and we will analyze the different cases separately.
Otherwise, there is a set of at most three that dominates $N_3(v, w) \cap V(R)$.
\lr{We will cover all the other vertices in the region by simple subregions induced by these vertices together with the poles $v$ and $w$.}

\begin{lemma}[Alber, Fellows, and Niedermeier~\cite{Alber2004}]\label{lemma:nonecover}
    Let $\fR$ be a maximal \dreg  for any semitotal dominating set $D$.
    For any $vw$-region $R \in \fR$, \revise{we have $\abs{N_1(v,w) \cap V(R)} \leq 4$} and the vertices lie on the boundary $\partial(R)$. 
\end{lemma}

\begin{lemma}[Garnero and Sau \cite{Garnero2018}]\label{lemma:ntwocover}\label{lemma:rtwosr}
    Given a reduced plane graph \G, a semitotal dominating set $D$ and a $\dreg$ $\fR$, then we have for each $vw$-region $R \in \fR$, $N_2(v,w) \cap V(R)$ can be covered by \revise{at most $6$} simple regions.
    \noindent Now assume that $G$ is only reduced under Rule~2.
    Given  a $vw$-region $R$, if $\Dv \neq \emptyset $ (resp. $\Dw \neq \emptyset$), then $N_3(v,w) \cap V(R)$ can be covered by: 
    \begin{enumerate*}[label=\allbold{(\alph*)}]
        \item $11$ \sr s~if $\Dw \neq \emptyset$ (resp. $\Dv \neq \emptyset$)\,,
        \item $14$ \sr s~if $N_{2,3}(v) \cap N_3(v,w) = \emptyset$.
    \end{enumerate*}
\end{lemma}

\begin{proof}
Observe that $\Dw \neq \emptyset$ triggers in Case 2 or 3 of Rule~2 which is the same as in \cite[Fact 6, arXiv v2]{Garnero2018}.\lrnote{We should contrast the difference to "reduced" in the Gernero-Sau paper and why our reductions imply the same.}
Note that this automatically includes our skipped case (both $\Dv \neq \emptyset$ and $\Dw \neq \emptyset$).
The second part of the claim follows by their arguments as well.
\end{proof}
\begin{lemma}[Region Size]\label{lemma:inside}
    Let $G$ be a reduced plane graph together with a given semitotal dominating set $D$.
   For a maximal \dreg $\fR$, any $vw$-region $R \in \fR$ contains at most $87$ vertices distinct from its poles.\lrnote{TODO: Check this proof agian.}
\end{lemma}
\begin{proof} 
    
    Observe that $\abs{N_2(v,w) \cap V(R)} \leq 6 \cdot 4 = 24$, because after reducing $N_2(v, w)$ can be covered by at most $6$ simple regions due to \cref{lemma:ntwocover}.
    By \cref{lemma:nonecover}, $\abs{N_1(v,w) \cap V(R)} \leq 4$.
    It remains to bound \nthi. 
    We will do so by distinguishing between each of the cases of \revise{Rule~2}. 
        Assume first that \revise{Rule~2} has \textbf{not} been applied in the following two cases, which \revise{can} only happen if either \revise{{\bf a) }$\Dvw \neq \emptyset$ or {\bf b) }$\Dvw = \emptyset$ and $(\Dv, \Dw\neq \emptyset)$.}
        We assume first that  $\Dvw \neq \emptyset$. Then there exists a set $\tilde D = \{d_1,d_2,d_3\} \subseteq \Dvw$ of at most three vertices dominating $N_3(v,w)$. 
        We observe that vertices from \revise{$N_3(v, w)\cap V(R)$}
        are common neighbors of either $v$ or $w$ (by the definition of a $vw$-region) and at least one vertex from $\tilde D$, because some vertex has to dominate them and we know that only the poles or vertices in $\tilde D$ come into question.
        Without violating planarity, we can span at most 6 distinct simple regions. Using the bound of simple regions and including $|\tilde D| = 3$, we can conclude $\nthi \leq 6 \cdot 4 + 3 = 27$.
        Now assume $\Dvw = \emptyset$, $\Dv \neq \emptyset$ and $\Dw \neq \emptyset$.
        Then \cref{lemma:rtwosr} implies that we can cover $R$ with at most $11$ simple regions giving us $\nthi {\leq 11} \cdot 4 = 44$ vertices.
        \revise{Note that Rule~2 has not changed the graph.}
        Now assume Case \ref{case:c2} has been applied: \revise{$N_2(v,w) \cap V(R)$} 
        was entirely removed and 
        \revise{$N_3(v,w) \cap V(R)$}
        replaced by four new vertices, i.e. $\nthi \leq 4$.
        Finally, if Cases \ref{case:c2} and \ref{case:c3} have been applied, $N_{2,3}(v) \cap N_3(v,w) \subseteq N_{2,3}(v)$ was replaced by one single vertex. 
        According to \cref{lemma:rtwosr}, we can cover $N_3(v,w) \setminus \{v'\} \cap V(R)$ with at most $14$ simple regions giving us $\nthi \leq 14 \cdot 4 + 1 = 57$.
        The other case is symmetric.
    \noindent Finally, we can bound $|V(R)| \leq 2 + 4 + 24 + \max(27, 44, 4, 57) = 87$, because $V(R) = \{v, w\} \cup \bigcup_i N_i(v,w)$.
\end{proof}

\subsubsection{Bounding the Vertices Not Enclosed by Any Region}\label{cpt:outside}
\noindent We continue by bounding the number of vertices that do not lie inside any region of a maximal \dreg $\fR$.
As we will see, these vertices are reduced by Rule~1.
We then ty to cover the rest with as few simple regions as possible.
In the following, we define \revise{$\dGR(v) = \abs{\{R(v,w) \in \fR: w \in D\}}$}, the number of regions in $\fR$ having $v$ as a pole \lr{which is captured by the degree of $v$ in the underlying multigraph \GR.}

\begin{lemma}[\cite{Alber2004}]\label{lemma:noneinside}
    Let $G$ be a plane graph, $D$ a semitotal dominating set, and $\fR$ a maximal \dreg of $G$. 
    If $u \in N_1(v)$ for \lr{every} $v \in D$ then $u \in V(\fR)$.
\end{lemma}

\begin{corollary}[\cite{Garnero2018}]\label{lemma:polesBound}
    Let \G be a graph and D be a set. 
    For any maximal \dreg~$\fR$ on $G$ it holds that $\sum_{v \in D} \dGR(v) = 2 \cdot \abs{\fR}$.
\end{corollary}

\begin{proposition}\label{lemma:outside}
    Let \G be a plane graph reduced under Rule~$1$ and $2$ and let $D$ be a semitotal dominating set of $G$. 
    For a maximal \dreg $\fR$,  we have
    \begin{align}
    \abs{V \setminus (V(\fR) \cup D)} \leq 97 \abs{D}\,.
    \end{align}
\end{proposition}
\begin{proof}
\lr{As in \cite{Garnero2019}, for every vertex $v\in D$, we have $N_1(v) \subseteq V(\fR)$. 
Hence no vertex lies outside.
}
After exhaustively applying the rules, we know that $|N_3(v)|\leq 1$ and  $N_2(v) \setminus V(\fR)$ can be covered by at most $4 \dGR(v)$ simple regions between $v$ and some vertices from $N_1(v)$ on the boundary of a region in $\fR$: 
    \begin{align*}
{\abs{\cup_{v \in D} N_2(v) \setminus V(\fR)}} \leq {4 \sum_{v \in D}4 \cdot \dGR(v)} = {16 \cdot \sum_{v \in D}\dGR(v)}  \leq  {32 \abs{\fR}}\,.
    \end{align*}
Summing up these three upper bounds for each $v \in D$ yields
\begin{align*}
\abs{V \setminus V(\fR) \cup D)} &\leq 32 \cdot \abs{\fR} + \abs{D}  \leq 32 \cdot (3 \abs{D} - 6) + \abs{D} \leq 97 \cdot \abs{D}\,,
\end{align*}
and proves the claim.
\end{proof}%
\noindent We are now ready to prove the main theorem.
\begin{proof}[Proof of~\cref{thm:central}]
    Let $G = (V, E)$ be the planar input graph together with a plane embedding.\lrnote{Connect with Corrollary 18.}
    Furthermore, let $G'=(V',E')$ be the graph obtained by the exhaustive application of \revise{all three reduction rules}.
     By \cref{lemma:correctnessone,lemma:correctnesstwo,lemma:correctnessthree} the procedure is sound and $G$ has a \sdom of size at most $k$ if and only if $G'$ has a semitotal dominating set of size at most $k$.
    Furthermore, by \cref{running-time}, it can be applied in polynomial time.
    Now, combining the previous results \cref{lemma:inside,lemma:outside,lemma:numRegions}, we obtain $V(G') = 87 \cdot (3k - 6) + 97 \cdot k  < \kernelsize \cdot k$ and if $\abs{V(G')} > \kernelsize \cdot k$, we simply reject the instance.
\end{proof}

\section{Other Classes of Graphs}\label{sec:graphclasses}

\noindent We prove all the claims in \cref{lemma:hardness} \revise{separately}.
We start with a reduction from \dom to \sdom that shows $W[2]$-hardness for bipartite graphs.
Henning and Pandey \cite{Henning2019} showed NP-hardness for chordal bipartite graphs, but their reduction is not an FPT reduction, as the parameter $k$ is blown up by a factor of $\cO(n)$.
\begin{lemma}
\lrnote{TODO: We should discuss how the best fix this bug.
 A reviewer pointed out a clear mistake in the proof: We need witnesses for $d_1,d_2$ as well in case when one part of the bipartition does not have dominating vertex. Easiest way might be to append another P2?}
    \sdom is $W[2]$-hard on bipartite graphs parameterized by solution size.
\end{lemma}
\begin{proof}
     We reduce from \dom on bipartite graphs parameterized by $k$, which is $W[2]$-hard~\cite[Theorem $1$]{Raman2008}.
    Let $(G,k)$ be an instance of \dom, where $G$ is bipartite with bipartition $(X,Y)$.
    Let $\WitRad$ be the witness radius as in \cref{def:witness-radius}.
 We construct a bipartite graph $G'$ as follows. 
 The \revise{main} idea is to introduce a vertex on both sides, where each provides a witness for one bipartition by forcing them into a solution by adding a pendant.
 Because there are cases where a solution only contains vertices from one side, we introduce a second one per side such that they witness each other.
 
Formally, for every $x\in X$ add a new vertex $a_x$ adjacent to $x$,
    and for every $y\in Y$ add a new vertex $b_y$ adjacent to $y$.
    Add four vertices $d_x$, $d_x'$ as well as $d_y$ and $d_y'$ such that both $d_x$ and $d_x'$ are adjacent to all $a_x$ and \revise{similarly connect} both $d_y$ and $d_y'$ to all $b_y$.
    Finally, \revise{link} each  $d_x$,$d'_x$,$d_y$ and $d'_y$ with one new vertex $p_x,p'_x$, $p_y$ and $p_y'$ and denote this as the \emph{pendant} set $P$.
    This construction is clearly bipartite.

    We will now show that $(G ,k)$ is a yes-instance for \dom if and only if $(G',k+4)$ is a yes-instance for \sdom.
    First, assume there exists a dominating set $D$ in $G$ of size $k$.
    Let $A := \{d_x,d_x', d_y, d_y'\}$.
    Then \revisenote{Removed "simply setting"} $D' = D \cup A$ is a semitotal dominating set in $G'$ of size $k+4$. 
    \revise{Note that the subgraph $G'[V(G') \setminus V(G)]$ is entirely dominated by the vertices in $A$ and the rest $G'[V(G)]$ is dominated as before.}
    Furthermore, every vertex is at a distance of at most two from some vertex in $A$ and by construction the vertices in $A$ witness each other, i.e. $d_x \in \WitRad(d_x')$ (and $d_y \in \WitRad(d_y')$).
    Therefore, any $d\in D$ has a witness in the reduced instance.
        
    For the other direction, assume there exists a semitotal dominating set $D'$ of size $k + 4$ in $G'$.
    We can assume that $D'$ does not contain any of the vertices in $p \in P$ because one could \revise{swap} them with their neighbor in $d \in A$ (together with the fact that $\WitRad(p) \subseteq \WitRad(d)$ preserving witness relationships).
    Hence, assume $A\subseteq D'$.
    Furthermore, no vertex in the set $\{a_x\mid x \in X\} \cup \{b_y \mid y \in Y\}$ is used, because \revise{swapping} them to any neighbor in $X$ or $Y$ will also preserve the witnesses.
    Then, we simply set $D := D' \setminus A$ which is clearly a dominating set in $G$ of size $k$ finishing the proof.
\end{proof}

In the \revise{\name{Colored Clique}} problem we are given a graph $G$ together with an assignment of colors to its vertices and an integer $k$.
The question is whether $G$ has a clique of size $k$ in which all the vertices are assigned distinct colors.
The next result follows directly from the reduction of~\cite{Bousquet2012} to show $W[1]$-hardness of \dom on circle graphs. 

\begin{corollary}
\sdom is $W[1]$-hard on \emph{circle} graphs parameterized by solution size.
\end{corollary}
\begin{proof}
Given an instance $(G,k)$ of \revise{\name{Colored Clique}}, the authors in~\cite{Bousquet2012} showed how to construct an equivalent instance $(G,k+{k \choose 2})$ of \dom on a circle graph $C$. 
They observe that, because of the way $C$ is constructed, \emph{every} dominating set of size $k + {k \choose 2}$ in $C$ induces a connected subgraph of $C$, and hence it is also a semitotal dominating set of $C$.  
\end{proof}

Recall that a graph $G$ is split if its vertex set can be partitioned into an independent set $S$ and a clique $K$ and that we call a vertex universal, if it is \revise{adjacent} to all other vertices in $G$.
First assume that $G$ is split \revise{and} does not contain a universal vertex. 
If $D$ is a \revise{minimum} dominating set, then there exists a set $D'\subseteq K$ that is also a dominating set of $G$ and $|D'|=|D|$
\revisenote{Thanks to both reviewers for pointing out that $D'$ could be smaller in the original phrasing. 

}.
Indeed, if $d\in D\cap S$, we can obtain a dominating set of the same size by removing $d$ and adding any neighbor $d'$ of $d$ \revise{because $d$ has no neighbor in $D \cap K$ by the minimality assumption.}

\revise{Since $D'\subseteq K$, we know that $D'$ induces a clique and moreover, as $G$ has no universal vertex, it holds that $|D'|\geq 2$. 
Hence every vertex of $D'$ has another vertex of $D'$ at distance one, and therefore $D'$ is a semitotal dominating set of $G$.
The case where there exists a universal vertex $u$ is trivial, as setting $D = \{u, s\}$ for any neighbor $N(v)$ yields a solution of constant size.
}
This gives the following hardness result.

\begin{lemma}
\sdom is $W[2]$-hard on \emph{split} graphs parameterized by solution size.
\end{lemma}

\noindent To the best of our knowledge, $W[2]$-hardness has not explicitly been shown for \tdom on bipartite graphs. 
We do so with the following reduction from \name{Set Cover} also used in~\cite{Chlebk2008}.

\begin{lemma}\label{lem:tds-bipartite}
\tdom is $W[2]$-hard for bipartite graphs.
\end{lemma}
\begin{proof}
We reduce from \name{Set Cover}.
Given an instance $\cI = ((U, \cF), k)$ of \revise{\name{Set Cover}} with universe $U$ and family $\cF$ of subsets of $U$ we create a bipartite graph $G$ with vertices $X := \{x_F \mid F \in \cF \}\cup \{ v'\}$ and $Y := U \cup \{v\}$.
Now, we connect each $x_F \in X$ to all $u \in Y$ for which $u \in F$ in $\cI$.
Additionally, we connect $v$ with all vertices in $X$.
Now we show that $G$ admits a total dominating set of size $k+1$ iff $\cI$ has a set cover of size $k$.

First, assume \revise{there exists a set cover} $C \subseteq \cF$ of size $k$. Each vertex $x_F \in X$ encodes a set and we simply put those into a solution that corresponds to the elements in $C$ together with $v$: By the fact that $C$ is a set cover, all vertices in $Y$ are dominated by at least one vertex in $X$.
Furthermore, $v$ dominates and witnesses all vertices in $X$, thus giving a total dominating set of size $k+1$.

Now, \revise{let $D$ be any total dominating set in $G$} of size $k + 1$\revise{.}
Note that we can assume $v \in D$ to dominate $v'$.
Now observe that \revise{if} some $\revise{u} \in Y$ is in $D$, we need a witness $\revise{x} \in X$ to be in $D$ as well.
Because each neighbor $n \in N(\revise{u})$ is already dominated by $v$, we can assume that a solution only consists of vertices in $X \cup \{v\}$.
This immediately gives a set cover of size $k$ as well in $\cI$.
\end{proof}

For any $t\in \mathbb{N}$, a graph $G$ is $t$-claw-free, if it does not contain an induced $K_{1,t}$ as a subgraph.
Note that a $t$-claw-free graph is also $t'$-free for all $t'> t$.

\begin{lemma}\label{lem:claw-free}
 \name{Total Domination} and \name{Semitotal Domination} are $W[2]$-hard for $t$-claw-free graphs for every integer $t \geq 4$.
\end{lemma}
\begin{proof}
We observe that the construction in \cite{Cygan2011} also holds for a semitotal dominating set.
We define the $W[2]$-hard problem \name{Colorful Red-Blue Dominating Set} as follows: For a given bipartite graph $G = (R \uplus B, E)$, parameter $k \in \mathbb{N}$ and colors $c:V\rightarrow [k]$, we ask whether there is a dominating set $D\subseteq R$ with $k$ distinctly colored vertices. 
We construct a $t$-claw-free graph $G'$ by making each $R$ and $B$ a clique. 
Furthermore, for each color class $R_i = \{v \in R \mid c(v) = i\}$ in $R$, we add a new vertex $v_i$ and connect it with all vertices in $R_i$.
Note that the diameter of this instance is exactly three; thus, every pair of distinct vertices witnesses each other.
It is now easy to see that any colorful red-blue dominating set in $G$ immediately gives a semitotal dominating set in $G'$.
Conversely, any semitotal dominating set $D$ must contain at least one vertex from each color class $R_i$ or $v_i$. 
In the latter case, we just swap it to some neighbor in $R_i$.
Observe that $D\subseteq R$, because otherwise \revise{it is impossible} to dominate all $v_i$'s for all $i \in [k]$.
\end{proof}

Note that the proof above does not work for $t=3$, that is, for claw-free graphs. %
In fact, \dom on claw-free graphs is fixed-parameter tractable, and it is an open question whether the same holds for \sdom and \name{Total Domination}.

\section{Acknowledgements}\label{sec:ack}
Retschmeier carried out this work at Basic Algorithms Research Copenhagen (BARC), which was supported by the VILLUM Foundation grant 54451.
Retschmeier is also supported by Providentia, a Data Science Distinguished Investigator grant from the Novo Nordisk Fonden.
Lima acknowledges support of the Independent Research Fund

\bibliographystyle{plainnat}

\bibliography{references-dblp}

\appendix
\newpage
\section{Complexity Landscape of \sdom}
\noindent With the following \cref{fig:bigpicture}, we supplement the treatment of the complexity state in \cref{tab:complexities} with an updated graph of the inclusions, similar to those shown in \cite{Galby2020}. 
In addition to the classic complexities, classes where an fpt algorithm (when parameterized by solution size) is known\revise{, are} marked in {\setulcolor{MATHAYELLOW}\ul{yellow}} (alternatively ``\textbf{*}'') and those where we know parameterized intractability in {\setulcolor{MATHARED}\ul{red}} (``\textbf{°}'').

\begin{figure}[t]
    \centering
    \resizebox{1.0\textwidth}{!}{
        \input{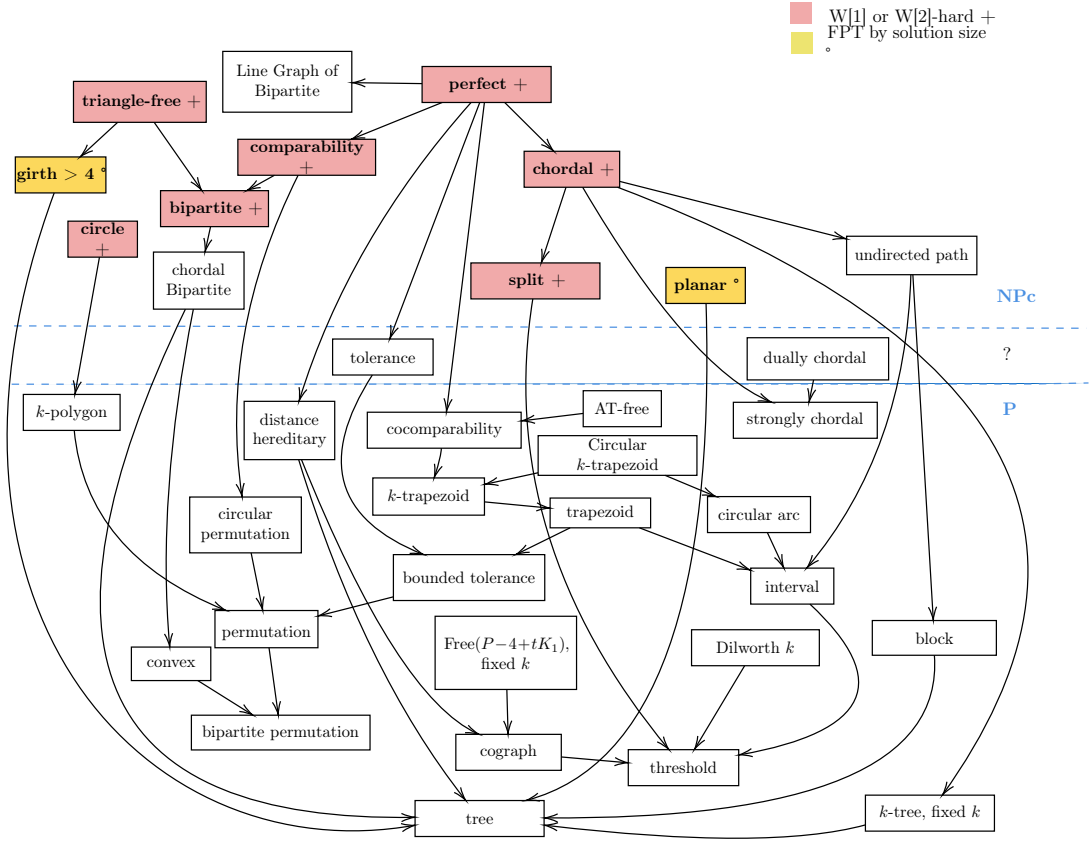}
    }
    \caption{Parameterized Complexity of \sdom. 
    The lower part shows classes where a poly-time algorithm is known. }
    \label{fig:bigpicture}
\end{figure}

\end{document}